\documentclass[11pt]{article}
\usepackage{amsfonts,amssymb,a4,bm}

\def\theequation{\thesection.\arabic{equation}}

\newcommand{\qed}{\hfill\rule{3mm}{3mm}}
\newtheorem{defi}{Definition}
\newtheorem{cor}{Corollary}
\newtheorem{lem}{Lemma}
\newtheorem{teo}{Theorem}

\newtheorem{pro}{Proposition}
\makeatletter \@addtoreset{equation}{section} \makeatother
\begin{document}


\voffset=-1.5truecm\hsize=16.5truecm    \vsize=24.truecm
\baselineskip=14pt plus0.1pt minus0.1pt \parindent=12pt
\lineskip=4pt\lineskiplimit=0.1pt      \parskip=0.1pt plus1pt

\def\ds{\displaystyle}\def\st{\scriptstyle}\def\sst{\scriptscriptstyle}


\let\a=\alpha \let\b=\beta \let\ch=\chi \let\d=\delta \let\e=\varepsilon
\let\f=\varphi \let\g=\gamma \let\h=\eta    \let\k=\kappa \let\l=\lambda
\let\m=\mu \let\n=\nu \let\o=\omega    \let\p=\pi \let\ph=\varphi
\let\r=\rho \let\s=\sigma \let\t=\tau \let\th=\vartheta
\let\y=\upsilon \let\x=\xi \let\z=\zeta
\let\D=\Delta \let\F=\Phi \let\G=\Gamma \let\L=\Lambda \let\Th=\Theta
\let\O=\Omega \let\P=\Pi \let\Ps=\Psi \let\Si=\Sigma \let\X=\Xi
\let\Y=\Upsilon



\global\newcount\numsec\global\newcount\numfor
\gdef\profonditastruttura{\dp\strutbox}
\def\senondefinito#1{\expandafter\ifx\csname#1\endcsname\relax}
\def\SIA #1,#2,#3 {\senondefinito{#1#2}
\expandafter\xdef\csname #1#2\endcsname{#3} \else
\write16{???? il simbolo #2 e' gia' stato definito !!!!} \fi}
\def\etichetta(#1){(\veroparagrafo.\veraformula)
\SIA e,#1,(\veroparagrafo.\veraformula)
 \global\advance\numfor by 1
 \write16{ EQ \equ(#1) ha simbolo #1 }}
\def\etichettaa(#1){(A\veroparagrafo.\veraformula)
 \SIA e,#1,(A\veroparagrafo.\veraformula)
 \global\advance\numfor by 1\write16{ EQ \equ(#1) ha simbolo #1 }}
\def\BOZZA{\def\alato(##1){
 {\vtop to \profonditastruttura{\baselineskip
 \profonditastruttura\vss
 \rlap{\kern-\hsize\kern-1.2truecm{$\scriptstyle##1$}}}}}}
\def\alato(#1){}
\def\veroparagrafo{\number\numsec}\def\veraformula{\number\numfor}
\def\Eq(#1){\eqno{\etichetta(#1)\alato(#1)}}
\def\eq(#1){\etichetta(#1)\alato(#1)}
\def\Eqa(#1){\eqno{\etichettaa(#1)\alato(#1)}}
\def\eqa(#1){\etichettaa(#1)\alato(#1)}
\def\equ(#1){\senondefinito{e#1}$\clubsuit$#1\else\csname e#1\endcsname\fi}
\let\EQ=\Eq
\def\0{\emptyset}

\def\pp{{\bm p}}\def\pt{{\tilde{\bm p}}}


\let\dpr=\partial \def\V#1{\vec#1} \def\Dp{\V\dpr}
\def\oo{{\V\o}} \def\OO{{\V\O}} \def\uu{{\V\y}} \def\xxi{{\V \xi}}
\def\xx{{\V x}} \def\yy{{\V y}} \def\kk{{\V k}} \def\zz{{\V z}}
\def\rr{{\V r}} \def\zz{{\V z}} \def\ww{{\V w}}
\def\Fi{{\V \phi}}
\def\La{\Lambda}

\def\\{\noindent}
\let\io=\infty

\def\VU{{\mathbb{V}}}
\def\ED{{\mathbb{E}}}
\def\GI{{\mathbb{G}}}
\def\Tt{{\mathbb{T}}}
\def\C{\mathbb{C}}
\def\LL{{\cal L}}
\def\RR{{\cal R}}
\def\SS{{\cal S}}
\def\NN{{\cal M}}
\def\MM{{\cal M}}
\def\HH{{\cal H}}
\def\GG{{\cal G}}
\def\PP{{\cal P}}
\def\AA{{\cal A}}
\def\BB{{\cal B}}
\def\FF{{\cal F}}
\def\TT{{\cal T}}
\def\v{\vskip.1cm}
\def\vv{\vskip.2cm}
\def\gt{{\tilde\g}}
\def\E{{\mathcal E} }
\def\I{{\rm I}}
\def\0{\emptyset}
\def\xx{{\V x}} \def\yy{{\V y}} \def\kk{{\V k}} \def\zz{{\V z}}
\def\ba{\begin{array}}
\def\ea{\end{array}}  \def \eea {\end {eqnarray}}\def \bea {\begin {eqnarray}}

\def\tende#1{\vtop{\ialign{##\crcr\rightarrowfill\crcr
              \noalign{\kern-1pt\nointerlineskip}
              \hskip3.pt${\scriptstyle #1}$\hskip3.pt\crcr}}}
\def\otto{{\kern-1.truept\leftarrow\kern-5.truept\to\kern-1.truept}}
\def\arm{{}}
\font\bigfnt=cmbx10 scaled\magstep1

\newcommand{\card}[1]{\left|#1\right|}
\newcommand{\und}[1]{\underline{#1}}
\def\1{\rlap{\mbox{\small\rm 1}}\kern.15em 1}
\def\ind#1{\1_{\{#1\}}}
\def\bydef{:=}
\def\defby{=:}
\def\buildd#1#2{\mathrel{\mathop{\kern 0pt#1}\limits_{#2}}}
\def\card#1{\left|#1\right|}
\def\proof{\noindent{\bf Proof. }}
\def\qed{ \square}
\def\reff#1{(\ref{#1})}
\def\eee{{\rm e}}

\def\xto#1{\xrightarrow{#1}}

\title{On stable pair potentials with an attractive tail, remarks on two papers by A. G. Basuev.
{}}
\author{
Bernardo N. B. de Lima\footnote{Departamento de Matem{\'a}tica, Universidade Federal de Minas Gerais, Belo Horizonte-MG, Brazil},
Aldo Procacci $^{*}$ and Sergio Yuhjtman \footnote{Departamento de Matem\'atica, Universidad de Buenos Aires, Buenos Aires, Argentina}
\\
}
\maketitle
\def\xxx{{\bf \underline{x}}}
\def\be{\begin{equation}}
\def\ee{\end{equation}}
\vskip.5cm

\begin{abstract}
We revisit two old and apparently little known papers by Basuev
\cite{Ba1} \cite{Ba2} and show that the results contained there
yield strong improvements on current lower bounds of the
convergence radius of the Mayer series for continuous particle
systems interacting via a very large class of stable and tempered
potentials which includes the Lennard-Jones type potentials. In
particular we analyze the case of the classical Lennard-Jones gas
under the light of the Basuev scheme and, using also some new
results \cite{Y} on this model recently obtained by one of us,  we provide a
new lower bound  for the Mayer series convergence radius of the
classical Lennard-Jones gas which improves by a factor of the order
$10^5$ on  the current best lower bound recently  obtained in
\cite{LP}.

 \end{abstract}

 \numfor=1\numsec=1
\section{Introduction}

The possibility to obtain the  equation of state for a non-ideal
gas only from first principles, i.e. once given the microscopic interaction
between its molecules,  has been a subject of intense investigation by the mathematical physics community and
several rigorous results on this issue has been obtained, especially in the decade  of the sixties, concerning mainly  systems of classical particles
interacting via a pair potential.

\\In this context, Mayer,  \cite{May37,May42} and Mayer and Mayer \cite{MM}, in their   seminal works in the forties, were able to obtain an explicit
representation of the pressure of a non ideal gas whose particles interact via a pair potential. Such an expression, known nowadays
by the name of Mayer series, is a
formal power series in terms of the activity (a quantity trivially related to the chemical potential)
whose coefficients (called the Ursell coefficients)
can be  computed once given the microscopic pair potential.

\\The Mayer series
represents an explicit expression of the equation of state of a non ideal gas deduced from first principles (i.e. only from the
microscopic pair interaction between particles).  However the series was at that time only formal, in the sense that nothing could
be said about its convergence due to the intricate combinatorial structure of the Ursell coefficients.
Even so,  the Mayer series (truncated at some order) has been successfully used since then by chemists and physicist
to extract useful information about the behavior of real gases.
This was a strong motivation  to prove
its convergence, at least in the region of low density or high temperature, where every classical system of particles  is supposed to be
in the gas state. Such an achievement  would provide a  firm ground
to the interpretation of the Mayer series as the   rigorous and exact equation of state for any  non-ideal gas.
\newpage

\\It took more than twenty years after Mayer's results to obtain the first rigorous results on  the convergence of the Mayer series. The
first breakthrough,  due to
Groeneveld \cite{gro62},  was a proof of the convergence of the Mayer series in the low-density/high-temperature phase for
gases of particles that
interact via a purely repulsive pair potential. The latter was a quite dramatic restriction, which ruled out most, if not all,
pair potentials modeling realistic continuous particle systems in physics. However, just one year  later, Penrose \cite{Pe1, Pe2} and independently  Ruelle \cite{Ru1,Ru2} obtained
the astonishing result that the Mayer series of continuous systems of particles
was  actually  an analytic function of the temperature and activity in the low density/high temperature phase. The result holds
for an enormous  class of pair interactions (the so-called  stable and tempered pair potentials,
see ahead for the definitions), which included practically all known examples of pair potentials for realistic gases.
More than this, Penrose and Ruelle were able to
provide a lower bound for the convergence radius of the Mayer series, which stands among the best available in the literature  till nowadays.

\\The mathematical methods developed by Penrose and Ruelle to get their results were  based on the so called Kirkwood-Salzburg equations
(see e.g.  Sec. 4.2.1 in \cite{Ru} and references therein). Such techniques
do not face directly the structure of the general $n^{\rm th}$ order Ursell
coefficient of the Mayer series (a well behaved upper bound for it, e.g., such as $C^n$ for some constant $C$,  would  provide
immediately the analyticity of the series  for low densities). Rather,  these methods
are based on the analysis of an infinite set of relations relating the $n$-point correlation function of the system
(which  can also be expressed as a series similar to the Mayer series) to the $m$-point correlation functions with $m<n$).

\\Despite the undoubting success of the methods based on Kirkwood-Salzburg equations,
 the apparently great difficulty to obtain a direct bound for the Mayer coefficients
starting from their explicit expression  in terms of sum over connected graphs remained an intriguing  open question.

\\The first result in this direction was obtained by  Penrose  in 1967 \cite{pen67} who was able to obtain bounds on the Mayer coefficients, as powerful as
those of Penrose-Ruelle in 1963 \cite{Pe2, Ru1},  directly via a resummation of the connected graphs appearing in the Ursell coefficients in terms of trees:  it was the
first example of tree graph identity. This very nice result, however, was restricted to pair potentials
which had not only to be stable and tempered but also possess a hard-core.
The Penrose tree-graph identity,  somehow forgotten for thirty years, was recently rescued in  \cite{FP},\cite{Pr1} and \cite{FPS}  where it has been utilized
 to improve the convergence region of the abstract polymer gas and the gas of hard spheres in the continuum respectively.

\\In 1978, Brydges  and Federbush  \cite{BF} developed a new tree-graph identity (different from that of Penrose) through which
they were able to obtain bounds stronger than those of Penrose-Ruelle for a subclass of the stable and tempered   potentials.
Pair potentials in this subclass  needed to be absolutely summable. This was  a quite severe restriction since it left out the
important class of the Lennard-Jones type potentials. However  a quite popular  potential among chemists was included:
the so called Morse Potential, and in this case the improvement was really consistent (see e.g. \cite{MPS}).

\\The original identity presented in \cite{BF} was successively reworked and extended
in several papers  (\cite{BaF}, \cite{bry84}, \cite{BK}, \cite{AR}, \cite{BM})
and its  quite robust structure has been successfully used mainly for applications in
constructive field theory, especially  during the eighties   (see e.g. \cite{GJ} and references therein).

\\There have been also more recent developments of the  Brydges-Federbush tree-graph identity which are directly related with the present paper
(see e.g. \cite{MPS}, \cite{Pr1} and \cite{PU}).
 In particular, via the Brydges-Federbush  identity, Morais et al. \cite{MPS}
were able to improve the Penrose-Ruelle region of analyticity of the Mayer series for (some particular cases of ) Lennard-Jones type potentials.
Based on the ideas of \cite{MPS} (and using also some results in \cite{LoS}) de Lima and Procacci  \cite{LP},
improved  strongly the convergence radius of the  classical  Lennard-Jones potential.

\\In 1978,  the same year of the original paper by  Brydges and Federbush, A. G. Basuev published
(in the Soviet Journal Theoreticheskaya i Mathematicheskaya Fizika) a paper \cite{Ba1} providing
a stability criterion  for a quite large class of pair potentials (which included the Lennard-Jones type potentials).
If such  criterion holds then
the potential can be split into a short distance non-negative compactly supported  term (the repulsive part) plus  an absolutely integrable term
(incorporating the possible attractive long range tail of the potential). The stunning  property of the absolutely integrable
part is that it has the same stability constant as the full potential.
This very interesting result was used by the same author one year later \cite{Ba2} to obtain bounds
for the Mayer  coefficients for this class of potentials (which, we repeat,  is very  large and possibly includes
nearly all reasonably physical example). These bounds  appeared to be  strongly  better than those obtained by Penrose and Ruelle.
Even more surprising, to get these bounds the author constructed an original tree graph identity,
different from that   of Penrose of 1967 \cite{pen67} as well as that of   Brydges-Federbush of 1978 \cite{BF}.
Apparently the author was not aware   at the time of the paper \cite{BF}, so he did not try to compare his result
with that of Brydges and Federbush. If he did,  he
would arrive at the conclusion that his methods yield (in the worst of the hypothesis) the same improvements given in \cite{BF}
for absolutely summable pair potentials
but of course are able to provide  fantastic bounds for non absolutely summable pair potential
which were at the time  out of reach with the Brydges-Federbush methods.

\\These two  works by Basuev, which should be cited in any paper dealing with the convergence of the Mayer series
for continuous systems, were instead nearly completely overlooked.
Actually the Basuev stability criterion presented in   \cite{Ba1} was cited in some detail once
\cite{PGM}. We  were not able to find any reference at all in the literature to the Basuev tree-graph identity
and the consequent improvements on the Mayer  coefficient bounds presented in the second paper \cite{Ba2}.

\\In the present work we revisit  these two papers by Basuev with the intent to rescue the important results there contained,
which apparently passed completely unobserved.
In particular, we show how these results yield improvements on all known results for a wide class of
pair potentials, focusing in particular on the important and physically relevant subclass of the Lennard Jones type potentials.
We also provide new results for the specific case of the classical Lennard-Jones potential $V(r)=1/r^{12}-2/r^6$ which strongly improve on the  results
given by  two  of us  in a recent paper \cite{LP}.  To get this last result on the specific case of the Lennard-Jones potential, we
used the very recent results obtained by one of us \cite{Y} on the stability constant and minimal inter-particle distance in lowest energy configurations
of the Lennard-Jones gas.

\\The rest of the paper is organized as follows. In section 2 we introduce notations and model. We recall the definition of
Mayer series and Mayer coefficients. We  introduce the concepts
of stability and   temperedness, and
the class of potentials satisfying the Basuev criterion (Definition \ref{def9}).
We finally recall the main results on Mayer series convergence by Penrose and Ruelle
(Theorem \ref{peru}),  and, the new recent results by Morais, Procacci and Scoppola (Theorem \ref{teo2}).
In section 3 we (re)introduce  the Basuev tree graph identity (Theorem \ref{basu0}) which we (re)prove in Appendix A. In section 4 we (re)state the
Basuev criterion (Theorem \ref{basu1}), whose proof is given in Appendix B, and (re)derive  from the Basuev tree graph identity
the bounds on the Ursell coefficient for particle systems interacting via a Basuev pair potential. Finally, in Section 5 we present the estimates of the Mayer coefficients
and convergence radius for Lennard-Jones type potentials and classical Lennard-Jones potential, showing how this estimates improve on the recent bounds
obtained in \cite{MPS} and \cite{LP}.

\section{ The model}\label{sec2}
\numsec=2\numfor=1
We  consider a continuous system of  particles in the $d$-dimensional Euclidean space $ \mathbb{R}^d$. Denote by $x_i\in \mathbb{R}^d$
the position vector of the $i^{\rm th}$ particle of the system and  by $|x_i|$ its Euclidean norm. We  suppose hereafter that particles interact through a
distance dependent pair potential $V(r)$ with $r\in [0,+\infty)$. Given
a configuration
$(x_1,\dots,x_n)\in {\mathbb R}^{dn}$ of the system such that  $n$ particles are present, the energy $U(x_1,\dots,x_n)$
of this configuration is defined as
$$
U(x_1,\dots,x_n)=\sum_{1\le i< j\le n}V(|x_i-x_j|)
$$
The statistical mechanics of the system is  governed by the following partition function in the Grand Canonical Ensemble
$$
\Xi_{\La}(\l,\b)=\sum_{n=0}^{\infty}{\l^{n}\over n!} \int_\L dx_1
\dots \int_\L dx_n e^{-\b U(x_1,\dots,x_n)}
\Eq(1.1)$$
where $\L\subset{\mathbb R}^d$ is the region where the system is confined (which  hereafter we assume to be  a $d$-dimensional
cube with sides of length $L$ and  denote by $|\L|$ its volume), $\l$ is the activity of the system
and $\b$ is the inverse temperature.

\\Thermodynamics is recovered by taking the logarithm of the partition function. Namely, the pressure  $P_\L$ and the density $\r_\L$ of the system at fixed values of the
thermodynamic parameters inverse temperature $\b$, fugacity  $\l$ and volume $\L$, are given respectively by the following formulas
$$
P_\L= {1\over \b |\L|}\log \Xi_{\La}(\l,\b)\Eq(pres)
$$
$$
\r_\L= {\l\over |\L|}{\partial\over \partial\l}\log \Xi_{\La}(\l,\b)\Eq(dens)
$$
As we will see below, both $P_\L$ and $\r_\L$ are given in terms of a power series of the fugacity $\l$ (with coefficients depending on $\b$).
Inverting \equ(dens)
by  writing $\l$ as a power series of the density $\r$ and plugging this series into \equ(pres) one obtains the  so called Viral series,
i.e. the pressure as a function of the density and temperature, or, in other words,  the equation of state of a gas whose particles
interact via the pair potential $V$.

\\The dependence of the pressure $P_\L$  and density $\r_\L$ on the volume $\L$
should be a residual one. In fact, one  may think to increase the
volume $\L$ of the  system   keeping fixed the value of the
fugacity $\l$ and the inverse temperature $\b$. It is  then expected that pressure and  density of the system do not vary
significantly. It is usual to  think that the
box $\L$ can be made  arbitrarily large (which is the rigorous formalization
of macroscopically large) and define
$$
P ~=~\lim_{\L\to \infty}{1\over \b |\L|}\log \Xi_{\La}(\l,\b)\Eq(1.31)
$$
$$
\r~=~\lim_{\L\to\infty}{1\over |\L|}{\l}{\partial\over \dpr
\l}\log \Xi_{\La}(\l,\b)\Eq(1.32)
$$
The limit $\L\to\infty$  (here $\L\to\infty$  means that the side $L$ of the cubic box $\L$ goes to  infinity)
is usually called by physicists the
{\it thermodynamic limit} and
the exact thermodynamic behavior of the system is in principle recovered at the
thermodynamic limit.

\\It is thus crucial  for the study of  the thermodynamic behavior of the system to  have available
 an explicit expression for $|\L|^{-1}\log \Xi_{\La}(\l,\b)$ and to be able to control its behavior  as $\L\to \infty$.

\\A very well known and old result  due to Mayer and Mayer \cite{MM} (see there chapter 13, p. 277-284)
states that
the quantity  $\log \Xi_\L(\b,\l)$, with  $\Xi_\L(\b,\l)$ being the partition function defined as in \equ(1.1), can be written in terms of a formal series in powers of $\l$. Namely,
$$
\ln \Xi_{\La}(\l,\b) ~=~ |\L|\left[\l+ \sum_{n=2}^{\infty}C_n(\b,\L)\l^n\right]\Eq(pressm)
$$
where
$$
C_n(\b,\L)~=~{1\over |\L|}{1\over n!}\int_{\L}\,dx_1
\dots \int_{\L} dx_n\: \sum\limits_{g\in G_{n}}
\prod\limits_{\{i,j\}\in E_g}\left[  e^{ -\b V(|x_i-x_j|)|} -1\right] \Eq(urse)
$$
with $G_n$ denoting the set of all connected graphs with vertex set
$[n]\equiv \{1,2,\dots,n\}$ and $E_g$ denoting the edge set  of $g\in G_n$.

\\The power series \equ(pressm) is the famous Mayer series, the factors $C_n(\b,\L)$ are usually  called {\it Mayer (or Ursell) coefficients} and the factors $\phi_\b(x_1,\dots,x_n)$
are usually called the
 {\it Ursell functions}  (see e.g. \cite{Ru} and references therein).

\\In order for the limits \equ(1.31) and \equ(1.32) to exist, the dependence on the volume $|\L|$
of the Mayer  coefficients $C_n(\b,\L)$ given by \equ(urse) must be
only  a residual one. Indeed,  via the expression \equ(urse) it is not difficult to show that
$C_n(\b,\L)$  is bounded   uniformly in $\L$ and, for every $n\in \mathbb{N}$,
the limit
$$
C_n(\b)=\lim_{|\L|\to\infty}C_n(\b,\L)\Eq(bninf)
$$
exists and it is a finite constant (see e.g \cite{Ru}).  On the other hand
it is a very difficult task  to obtain an upper  bound on $|C_n(\b,\L)|$  which is  at the same time  uniform in the volume $|\L|$ and well-behaved
in $n$:  i.e. an upper bound of the form
$$
|C_n(\b,\L)|\lesssim (C_\b)^n\Eq(goodcomb)
$$
with $C_\b$ some positive constant. Such an estimate
would immediately
yield a lower bound for the convergence radius of the
Mayer series of the pressure uniform in the volume. The difficulty to obtain a bound of type \equ(goodcomb)
directly from the expression  \equ(urse) of the Mayer coefficients stems from the fact
 that the cardinality of the set  $G_n$ (the connected graphs  with vertex set $[n]$) behaves very badly with $n$
(it is no less than $2^{(n-1)(n-2)/2}$, which is the number of connected graphs containing a fixed tree).
Therefore,
it seems more than reasonable to impose some conditions on the potential $V$
and at the same time it seems inevitable to look for   some  hidden cancelations  in the
expression \equ(urse).

\\During the sixties it became clear that the minimal conditions one has to impose on the potential $V$
to get bounds for the Ursell coefficients of the type \equ(goodcomb) were basically {\it stability} and {\it temperedness}, whose definitions we recall below.

\begin{defi}[Stability]
A pair potential $V$ is said to be stable if there exists a finite $C\ge 0$ such that, for all ${n}\in \mathbb{N}$ and all $(x_1,\dots,x_n)\in \mathbb{R}^{dn}$,
$$
U(x_1,\dots,x_n)=\sum_{1\le i< j\le n}V(|x_i-x_j|)\ge -n C \Eq(stabi)
$$
\end{defi}

\begin{defi}[Temperedness]
A pair potential $V$ is said to be tempered  if there exists $r_0\ge 0$ such that
$$
\int_{|x|\ge r_0} dx |V(|x|)| < \infty \Eq(temp)
$$
\end{defi}
Stability and temperedness
are actually deeply interconnected and the lack of one of them always produces non thermodynamic or catastrophic  behaviors (see e.g.
the  nice and clear exposition about the role of stability and temperedness  in   reference \cite{ga}).
As far as stability is concerned, observe that
the grand-canonical partition function defined in \equ(1.1) is  a holomorphic function of $\l$ if the potential $V$ is stable. Moreover, under
very mild additional conditions on the potential  (upper-semicontinuity) it can be proved that  the converse is also true (see Proposition 3.2.2 in \cite{Ru}). In other words
$\Xi_{\La}(\l,\b)$ converges
if and only if the potential $V$ is stable. So, in some sense, stability is a condition {\it sine qua non} to construct a    consistent  statistical mechanics
for continuous neutral particle systems\footnote{For systems of charged particles
stability is no longer a necessary condition, see e.g. the example of the two-dimensional Yukawa gas in \cite{BK} or \cite{BM}.}.
Concerning   temperedness, while it is reasonable to expect that the potential  decays to zero at
large distances, on the other side it is not completely evident to see why the potential should decay
in an absolutely integrable way. E.g., there are examples of charged
particle systems, such as the dipole gas, which interact through a non-tempered pair potential and have a convergent Mayer series. We denote
$$
B_n=\sup_{(x_1,\dots,x_n)\in \mathbb{R}^{dn}}-{1\over n}U(x_1,\dots,x_n)\Eq(bn)
$$
$$
B=\sup_{n\ge 2} B_n\Eq(b)
\mathrm{}$$
We also define, for later use
$$
\overline{B}_n= {n\over n-1} B_n~~~~~~~~~\overline{B}=\sup_{n\ge 2}\overline{B}_n\Eq(barbn)
$$
\\As
will see below, the quantities $B$
and $\overline{B}$
play a crucial role in the estimates of the convergence radius of the Mayer series. Note  that  $B$, when finite,  is actually the best constant $C$ in \equ(stabi).
Note also that temperedness of $V$ implies that  $B$ is non-negative  and $B=0$ if and only if $V\ge 0$ (i.e. ``repulsive" potential).
The non-negative number $B$ defined in \equ(b), when finite, is known as   the {\it stability constant of the potential $V$}. The slight variant
$\overline{B}$ defined in \equ(barbn) has been originally introduced in \cite{Ba1}.
Of course, by definition $\overline{B}\ge B$ and
$B$ and $\overline{B}$, if not coinciding, should be  very close for
most of the reasonable stable pair potentials.  As noted by Basuev \cite{Ba1}, it holds that $B\le \overline{B}\le {13\over 12}B$
for any stable and tempered three-dimensional pair potential
which is definitively negative at sufficiently large distances.

\\In the present paper we work with a proper subclass of the tempered and stable pair potentials which was first proposed  by Basuev in \cite{Ba1}.
This class  is
sufficiently large to embrace practically every known example of physically relevant pair potential and we will see in Section \ref{sec4}
that a potential in this  class
has  a key  property (Theorem \ref{basu1} ahead) which is the
main ingredient  toward a very  effective bound of the form \equ(goodcomb) for the Ursell coefficient \equ(urse).
The class  is defined as follows.
\begin{defi}\label{def9}
A tempered pair potential  $V$  is called
 {\it Basuev} if there exists $a>0$ such that
$$
V(r)\ge V(a)>0 ~~~~~~~~~~~~~~~~~~~~~~~~~{\rm for ~all}~~~ r\le a\Eq(cb0)
$$
and
$$
V(a)> 2\mu(a)\Eq(cba)
$$
with
$$
\mu(a)= \sup_{n\in \mathbb{N},\,(x_1,\dots, x_n)\in \mathbb{R}^{dn}\atop |x_i-x_j|>a~ \forall \{i,j\}\in E_n}\sum _{i=1}^n
V^-(|x_i|)
\Eq(mua)
$$
where $V^-$ is the negative part of the potential, i.e. $V^-={1\over 2}\Big[|V|-V\Big]$
\end{defi}

\\As mentioned in the introduction,
the best rigorous bound on $|C_n(\b,\L)|$ available in the literature  for
stable and tempered pair  potentials is that
given by Penrose and Ruelle in 1963 \cite{Pe1,Pe2,Ru1,Ru2}.
\begin{teo}[Penrose-Ruelle]\label{peru}
 Let $V$ be a stable and tempered pair  potential with  stability constant $B$. Then
the  $n$-order Mayer  coefficient $C_n(\b,\L)$ defined in \equ(urse)
is bounded by
$$
|C_n(\b,\L)|\le e^{2\b B (n-1)}n^{n-2} {[C(\b)]^{n-1}\over n!}\Eq(bmaru)
$$
where
$$
C(\b)=\int_{\mathbb{R}^{d}} dx ~ |e^{-\b V(|x|)}-1|
$$
Therefore the Mayer series \equ(pressm) converges absolutely, uniformly in $\L$
 for any complex  $\l$ inside the disk
$$
|\l| <{1\over e^{2\b B+1} C(\b)}\Eq(radm)
$$
\end{teo}

\\The Penrose-Ruelle  bound has been improved  for some restricted classes of stable and tempered pair potentials.
Brydges and Federbush \cite{BF} gave an improvement of  the Penrose-Ruelle bound
as far as  absolutely summable pair potentials are considered and it can be seen that such improvement is quite strong
for several physically relevant pair potentials (see e.g. the example of the Morse potential treated in \cite{MPS}).
Recently, via an extension of the Brydges-Federbush identity, Morais et al.  \cite{MPS}
improved  the Penrose-Ruelle bound, by considering pair potentials
that can be written as a sum of a positive part plus an absolutely integrable stable part.

\begin{teo}[Morais-Procacci-Scoppola]\label{teo2}
Let $V$ be a stable and tempered pair potential with stability constant $B$. Let $V$ be
such that $V=\Phi_1+\Phi_2$, with $\Phi_1\ge 0$ and $\Phi_2$  absolutely integrable and stable
 with  stability constant  $\tilde B$.
 Then the $n$-th order Mayer coefficient $C_n(\b,\L)$ defined in \equ(urse)
admits the bound
$$
|C_n(\b,\L)|\le e^{{\b \tilde B}n}~n^{n-2} {[\tilde C(\b)]^{n-1}\over n!}\Eq(bteo2)
$$
where
$$
\tilde C(\b) =\int dx \left[ |e ^{-\b \Phi_1(|x|)} -1|+ \b|\Phi_2(|x|)|\right] \Eq(tilb)
$$
Consequently, the Mayer series \equ(pressm) converges absolutely for all complex activities $\l$  such  that
$$
|\l|<{1\over e^{{\b \tilde B}+1} \tilde C(\b)}\Eq(radmru)
$$
\end{teo}

 \vv
\\{\bf Remark}. We draw the attention of the reader to the fact that in expression \equ(radmru) for the lower bound of the convergence radius
there appears the stability constant $\tilde B$ of the
absolutely summable part $\Phi_2$ of the potential in place of the stability constant $B$ of the full potential $V$ as  in \equ(radm). This is somehow unpleasant
since  in general
$\tilde B\ge B$, so that bound \equ(radmru) has chance to be an improvement of the Penrose-Ruelle bound
\equ(radm) only if $\tilde B< 2B$.
Actually in \cite{MPS} the authors did produce some examples of pair potentials
that can be decomposed as above, with $\tilde B=B$. 
Moreover
in \cite{LP} it is shown that also the classical Lennard-Jones potential (see ahead for the definition) has this property.
As we will, all these examples  fall in the class of the Basuev potentials according to Definition
\ref{def9} and are thus covered by Theorem \ref{teop}  in section 4 below.


\section{The Basuev tree-graph identity}\label{secbas}

\numsec=3\numfor=1
In this section we present the Basuev tree-graph identity. As the Penrose identity (formula (6) in \cite{pen67}) and the Brydges-Federbush identity
(formula (50) in \cite{BF}),
it is an alternative expression of the Ursell coefficients \equ(urse) in
terms of a sum over trees rather than connected graphs.  As mentioned in the introduction and as we
will see ahead with all details,  this  identity  permits to get rid of the combinatorial problem
by directly bounding the Ursell coefficients  starting from their explicit expressions \equ(urse). The advantage of the Basuev tree-graph identity given in  \cite{Ba2} respect to the previous
ones given by Penrose \cite{pen67} and by Brydges and Federbush \cite{BF}  is twofold: it permits to treat a much larger
class of pair potentials and it permits to obtain better bounds.
It is worth to say,
that the Basuev identity can be deduced  as a particular case of
most recent versions \cite{BK,AR,BM} of the Brydges-Federbush tree-graph identity, e. g. such as that presented in Theorem VIII.3 of \cite{BM}.

\\The Basuev formula is essentially algebraic and combinatorial. In order
to introduce it, we need to give some preliminary notations.
\def\E{{\rm E}}
\vv

\\Given a set $X$, we denote by ${\rm P}(X)$ the set of subsets of  $X$ (i.e. ${\rm P}(X)$ is the power set of $X$) and
by ${\rm P}^*(X)$ the set of non-empty subsets of  $X$ (i.e. ${\rm P}^*(X)= {\rm P}(X)\setminus\{\0\}$).
\vv
\\{\it Set partitions}. Given a finite set $X$, we denote by $\Pi(X)$ the
set of all partitions of $X$. Namely, an element $\pi\in \Pi(X)$ is, for $k=1,\dots |X|$, a collection
$\p=\{\a_1,\dots, \a_k\}$  of non-empty pairwise disjoint subsets of $X$ such that $\cup_{i=1}^k\a_i=X$.
The elements $\a_i$ ($i=1,\dots, k$) of the collection forming $\pi$
are called {\it the blocks} of $\pi$.
Given a partition $\p=\{\a_1,\dots, \a_k\}$ of $X$ we set $|\pi|=k$.
For fixed $k=1,\dots, |X|$ we  denote by $\Pi_k(X)$ the set of all partitions
of $X$  into exactly $k$ blocks, or, in other words,   $\Pi_k(X)$ is the set of all $\pi\in \Pi(X)$ such that $|\pi|=k$.

\vv\vv
\\Let $n\in \mathbb{N}$. We remind the notation
$[n]=\{1,2,\dots,n\}$. We also denote by ${\rm E}_n$ the set of all
unordered pairs in $[n]$.
\begin{defi}\label{pair}
A pair interaction in the set  $[n]$ is a map ${V}: \E_n\rightarrow
\mathbb{R}\cup\{+\infty\}$ that associates to any unordered pair in
$\{i,j\}\in E_n$ a number $V_{ij}= V(\{i,j\})= V_{ji}$  with values in $
\mathbb{R}\cup\{+\infty\}$.
\end{defi}


\begin{defi}\label{algpot}
Given a pair interaction $V$ in $[n]$, the energy induced by $V$
is a map $U: {\rm P}([n])\rightarrow \mathbb{R}$ defined as follows. For
any subset $X$ of $[n]$
$$
U(X)~=~\sum_{\{i,j\}\subset X}V_{ij}\Eq(UX)
$$
\end{defi}

\begin{defi}\label{gibfact} Let $\b\in (0,+\infty)$.
The Gibbsian factor of the potential $U$ is the map $e^{-\b U}:{\rm P}([n])\to\mathbb{R}: X\mapsto e^{- U(X)}$.
\end{defi}

\begin{defi}\label{ursell}
The  { connected function or Ursell coefficient} (see e.g. \cite{BM}, \cite{Ru} and references therein) is a
map $\phi_\b:{\rm P}^*([n])\to\mathbb{R}: X\mapsto \phi_\b(X)$, where
$ \phi_\b(X)$, for all non-empty $X\subset [n]$,  is defined (uniquely and recursively) by the  equations
$$
e^{-\b U(X)} ~=~\sum_{\pi\in \P(X)}\prod_{\a\in \pi}\phi_\b(\a)
\Eq(2.1s)
$$
\end{defi}

\begin{lem}\label{conne} The following identities hold:
$$
\phi_\b(X)~=~\sum_{g\in {G}_{X}}\prod_{\{i,j\}\in
E_g}\left(e^{- \b V_{ij}}-1\right)\Eq(2.4s)
$$
where $G_X$ denote the set of all connected graphs with vertex set $X$ and for $g\in G_X$, $E_g$ is the set of edges of $g$.
\end{lem}

\\{\bf Proof}.
Recalling \equ(UX), it is possible to
invert the set of equations \equ(2.1s), i.e. to write down
explicitly the Ursell coefficients. This is achieved by a so called Mayer
expansion on the Gibbsian factor. In fact, for all finite  $X\subset
\mathbb{N}$ we can write
$$
e^{-\b U(X)}~=~\prod_{\{i,j\}\subset X}[(e^{- \b V_{ij}}-1)+1]~=~ \sum_{g\in
\GG_{X}}\prod_{\{i,j\}\in E_g}(e^{- \b V_{ij}}-1)$$
where $\GG_X$ denote the set of all graphs (connected or not connected) with vertex set $X$.
Collecting
together the connected components in the sum $\sum_{g\in \GG_{X}}$ we get
$$
e^{- \b U(X)}~=~
\sum_{\p\in \P(X)}\prod_{\a\in \pi} K_\b(\a)
\Eq(2.2s)
$$
with
$$
K_\b(\a)~=~\sum\limits_{g\in {G}_{\a}}\prod\limits_{\{i,j\}\in
E_g}\left(e^{-\b V_{ij}}-1\right) \Eq(2.3s)
$$
Comparing \equ(2.1s) with \equ(2.2s) we get $K_\b(X)~=\phi_\b(X)$, for all $X\subset [n]$.
$\Box$
\vv
\\{\bf Remark} Note that
$\phi_\b(X)=1$ when $|X|=1$ since the sum over $G_X$ contains just one graph with no edges and the empty product is equal to one.

\vskip.2cm

\\Let $\GG_n$ be the set of graphs (connected or not connected) with vertex set $[n]$. A graph $g\in \GG_n$
induces a partition $\p(g)\in \Pi([n])$
whose blocks are the sets of vertices of the connected components of $g$.

\\We
denote by $T_n$ the set of all trees with vertex set $[n]$ (connected graphs in $[n]$ with no loops, or also, connected graphs in $[n]$ with exactly
$n-1$ edges). For a fixed tree $\t\in T_n$ with edge set $E_\t$, we  denote by $\Phi_\t$ the set of all bijections $\ph_\t: [n-1]\to E_\t$
(the set of all labelings of the edges of $\t$). Given a tree $\t\in T_n$ and an edge labeling $\ph_\t\in \Phi_\t$ we denote  by
$\t_k$ the graph with vertex set $[n]$ and edge set $E_{\t_k}=\{\{i,j\}\in E_\t: \ph_\t^{-1}(\{i,j\})\le k\}$
(in other words $\t_k$ is the graph in $\GG_n$ obtained from  $\t$ by erasing all edges with labels greater than $k$). 

\begin{teo}[Basuev 1979]\label{basu0}
The following identity holds
$$
\phi_\b([n])~=~\sum_{g\in G_n}\prod_{\{i,j\}\in E_g} \left[e^{-\b
V_{ij}}-1\right]~=~~~~~~~~~~~~~~~~~~~~~~~~~~~~~~~~~~~~~~~~~~~~~~~~~~~~~~
$$
$$
= {(-1)^{n-1}}\int_0^\b d\b_1\int_0^{\b_1}d\b_2\cdots
\int_0^{\b_{n-2}}d\b_{n-1}\sum_{\t\in T_n}\prod_{\{i,j\}\in E_\t}
V_{ij}~\sum_{\ph_\t\in \Phi_\t} e^{
-\sum\limits_{k=1}^{n-1}(\b_{k}-\b_{k+1})\sum_{\a\in \pi({\bm
\t}_k)}U(\a)} \Eq(treeb)
$$
where $\b_n=0$ by convention.
\end{teo}

\vv
\\{\bf Remark }. In the case of pair potentials with a hard-core  (e.g. $V_{ij}=+\infty$ for some $\{i,j\}$),
one has just to  introduce a cut-off $H>0$
to
give meaning to the expression in the r.h.s. of \equ(treeb). One  defines an auxiliary
potential $V_{ij}(H)$ such that  $V_{ij}(H)=V_{ij}$ if $V_{ij}<+\infty$ and
$V_{ij}(H)=H$ if $V_{ij}=+\infty$. Both sides of the identity \equ(treeb) are well defined as long as the  pair potential $V(H)$ is considered. One then   can
just take the limit $H\to \infty$ to recover the identity \equ(treeb) for the hard-core potential.

\vv\vv

\\In the original paper \cite{Ba2}  Basuev did
not explicitly state the algebraic identity \equ(treeb).  For this reason, and also to make this paper as self-contained as possible,
we present in Appendix A a detailed proof of the identity \equ(treeb). Of course,
all ideas contained in the proof  are due to Basuev.

\numsec=4\numfor=1

\section{Estimates for the Mayer coefficients}\label{sec4}
In this section we show how to use the Basuev tree graph identity \equ(treeb) introduced in the previous section  to get
direct and ``combinatorially good" (in the sense of \equ(goodcomb))
upper bounds on the  Ursell coefficients \equ(urse). We will see in Section \ref{sec5} that such  bounds
are better than all known previous bounds (from the  old Penrose-Ruelle bound \equ(bmaru) to the recent Morais-Procacci-Scoppola bound \equ(bteo2)) as long as
we  restrict ourselves  to the   class of stable and tempered potentials introduced  in Definition \ref{def9}
which, we recall, embraces nearly all  relevant examples of pair potentials in physics.
As previously said, a
Basuev pair potential according to Definition \ref{def9} has  a property which is crucial in order to obtain these bounds. This property
is expressed by the following theorem.
\begin{teo}[Basuev 1978]\label{basu1}
Let $V$ be a Basuev potential according to Definition \ref{def9},  then  $V$ is stable. Moreover the representation
$$
V= V_a + K_a
$$
with
$$
V_a(r)=\cases{V(a) & if  $r\le a$\cr\cr
V(r)  & if $r>a$
}\Eq(va)
$$
and
$$
K_a(r)=\cases{V(r)-V(a) & if $r\le a$\cr\cr
0 &if $r>a$
}\Eq(ka)
$$
is such that
the potential
$V_a$ defined in \equ(va) is also stable and it has the same stability constants $B$ and $\overline{B}$ of the full potential  $V$.
\end{teo}
Once again, to make the present paper as self-contained as possible, we give in Appendix B the original proof of Theorem \ref{basu1}
due to Basuev.

\\This feature of a  Basuev potential $V$  (i.e. $V=K_a+V_a$ with $K_a$ positive compactly supported and $V_a$
stable, tempered and absolutely integrable, with stability constants equal to the ones of $V$) matches very well with the simple and handy structure
of the exponential factor $\exp\{
-\sum_{k=1}^{n-1}(\b_{k}-\b_{k+1})\sum_{\a\in \pi({\bm
\t}_k)}U(\a)\} $ appearing in the  Basuev tree graph identity \equ(treeb).  It allows to use the condition of stability \equ(stabi) to bound the part
in the exponent coming from the term $V_a$. At the same time, the positive term  $K_a$ in the exponent controls the
eventual divergence of the full potential at short distances.

\\Indeed, from Theorems \ref{basu0} and \ref{basu1}  the following estimate for the Ursell coefficients can be obtained.
Once again, the following theorem is a development of a result due to Basuev
\cite{Ba2}.

\begin{teo}[Basuev 1979]\label{teop}
Let $V$ be a Basuev pair potential with stability constant $B$.
Let $n\in \mathbb{N}$ such that $n\ge 2$. Let
$C_n(\b,\L)$ be defined as in \equ(urse) and let $\overline{B}$ be defined
as in  \equ(barbn). Then there exists $a\ge0$ such that the following upper bounds
 hold
\begin{itemize}
\item[1)]
$$
C_n(\b,\L)\le {n^{n-2}\over n!}\;e^{n\b B}\; \left[  C^*(\b)
\right]^{n-1}\Eq(A1)
$$
with
$$
C^*(\b)= \int_{|x|\le a} dx~\b\, |V(|x|)| {(1-e^{-\b
[V(|x|)-V(a)]})\over \b[V(|x|)-V(a)]} + \int_{|x|\ge a} dx~
\b\,|V(|x|)|\Eq(cstar)
$$

\item[2)]
$$
C_n(\b,\L) \le {n^{n-2}\over n!}\left[{e^{\b\overline{B}}-1\over \b\overline{B}}\right]^{n-1}\left[\hat{C}(\b, \overline{B})
\right]^{n-1}\Eq(A2)
$$
with
$$
\hat{C}(\b, \overline{B})= \int_{|x|\le a} dx ~ \b\,|V(|x|)|
{\b\overline{B}(1-e^{-\b [V(|x|)-V(a)-\overline{B}]})\over \b(V(|x|)-V(a)-\overline{B})(e^{\b\overline{B}}-1)}+ \int_{|x|\ge a} dx ~ \b\,|V(|x|)|\Eq(chat)
$$
\end{itemize}
Therefore the series \equ(pressm) converges absolutely, uniformly in $\L$
 for any complex  $\l$ inside the disk
$$
|\l| <\max\left\{{1\over e^{\b B+1}  C^*(\b)}~~; ~~{\b \overline{B}\over e(e^{\b\overline{B}}-1)\hat{C}(\b, \overline{B})}\right\}\Eq(radba)
$$
\end{teo}
\\{\bf Remark}. The first inequality \equ(A1) can be deduced also by using the early version of
the Brydges-Federbush tree graph identity (e.g., the version presented in Theorem 3.1 in \cite{bry84} or in Theorem 2 in \cite{pdls}).
On the other hand, as said at the beginning of Section \ref{secbas}, the later versions of the Brydges-Federbush tree graph identity
(such as the one given in Theorem VIII.3 in \cite {BM}) could have been used
to obtain  the second inequality \equ(A2).

\vv\vv
\\{\bf Proof}.
We first prove 1) i.e. inequality \equ(A1). Let us denote shortly
$$
\phi_\b(x_1,\dots,x_n)=\sum\limits_{g\in G_{n}}
\prod\limits_{\{i,j\}\in E_g}\left[  e^{ -\b V(|x_i-x_j|)} -1\right]
$$
\\Starting from the Basuev tree graph identity  \equ(treeb) with $V_{ij}=V(|x_i-x_j|)$
 we have immediately the following upper bound
$$
|\phi_\b(x_1,\dots,x_n)|~ \le~
\int_0^\b d\b_1\int_0^{\b_1}d\b_2\cdots \int_0^{\b_{n-2}}d\b_{n-1}\sum_{\t\in T_n}\prod_{\{i,j\}\in E_\t} |V(|x_i-x_j|)|~\times
$$
$$
\times
\sum_{\ph_\t\in \Phi_\t}
e^{ -\sum\limits_{k=1}^{n-1}(\b_{k}-\b_{k+1})\sum_{\a\in \pi({ \t}_k)}U(\a)}~~~~~~~~~~~~~~~\Eq(ini)
$$
By hypothesis, since $V$ is Basuev, we have that $V=V_a+K_a$ where $V_a$ is stable  with the same stability constant
of $V$ and $K_a$ is positive and supported in $[0,a]$. Therefore, for any $\a\subset [n]$ we have that
$$
U(\a)= \sum_{\{i,j\}\subset \a} V(|x_i-x_j|)= U_a(\a) + \tilde U_a(\a)
$$
with
$$
U_a(\a)= \sum_{\{i,j\}\subset \a} {V_a}(|x_i-x_j|)
$$
and
$$
\tilde U_a(\a)= \sum_{\{i,j\}\subset \a} {K_a}(|x_i-x_j|)
$$
Now, since $V_a$ is stable with stability constant $B$ we have that, for all $\a\subset[n]$
$$
\sum_{\{i,j\}\subset\a} {V_a}(|x_i-x_j|)\ge -B|\a|
$$
Thus, for any partition $\p$ of $[n]$
$$
\sum_{\a\in \pi}U_a(\a)\ge -B\sum_{\a\in \pi}|\a|\ge -Bn
$$
So we get (recall that, for all $k\in[n-1]$ we have $\b_{k}-\b_{k+1}\ge 0$ and
$\sum_{k=1}^{n-1} (\b_k-\b_{k+1})=\b_1\le \b$)
$$
\exp\left\{ -\sum_{k=1}^{n-1}(\b_k-\b_{k+1})\sum_{\a\in \pi(\t_k)}U(\a)\right\} = \exp\left\{ -\sum_{k=1}^{n-1}(\b_k-\b_{k+1})\sum_{\a\in \pi(\t_k)}[U_a(\a)+\tilde U_a(\a)]\right\}\le
$$
$$
\le \exp\left\{ -\sum_{k=1}^{n-1}(\b_k-\b_{k+1})[-B n+\sum_{\a\in \pi(\t_k)}\tilde U_a(\a)]\right\}\le
e^{\b Bn} e^{ -\sum_{k=1}^{n-1}(\b_k-\b_{k+1})\sum_{\a\in \pi(\t_k)}\tilde U_a(\a)}
$$
We have now
$$
\sum_{k=1}^{n-1}(\b_k-\b_{k+1})\sum_{\a\in \pi(\t_k)}\tilde U_a(\a)~=\sum_{k=1}^{n-1}\b_k\Big[\sum_{\a\in \pi(\t_k)}\tilde U_a(\a)-\sum_{\a\in \pi(\t_{k-1})}\tilde U_a(\a)\Big]
$$
The factor $\sum_{\a\in \pi(\t_k)}\tilde U_a(\a)-\sum_{\a\in \pi(\t_{k-1})}\tilde U_a(\a)$ is a sum of positive terms
of the form $K_a(|x_s-x_l|)$ indexed by edges $\{s,l\}\in E_n$  including the edge   $\{i,j\}$  of $\t$  which has the label $k$, that is, such that $\ph_\t(\{i,j\})=k$. Therefore we get

$$
\sum_{k=1}^{n-1}(\b_k-\b_{k+1})\sum_{\a\in \pi(\t_k)}\tilde U_a(\a)~\ge
~\sum_{\{i,j\}\in E_\t}\b_{\ph_\t(\{i,j\})}{K_a}(|x_i-x_j|)\Eq(see)
$$
Hence
$$
\exp\left\{ -\sum_{k=1}^{n-1}(\b_k-\b_{k+1})\sum_{\a\in \pi(\t_k)}U(\a)\right\} \le e^{\b Bn}\prod_{\{i,j\}\in E_\t}e^{-\b_{\ph_\t(\{i,j\})}{K_a}(|x_i-x_j|)}
$$
Therefore we get
$$
|\phi_\b(x_1,\dots,x_n)|\le {e^{nB\beta}}
\int_{0}^\b d\b_1\int_0^{\b_1} d\b_2\dots \int_{0}^{\b_{n-2}}d\b_{n-1}~F(\b_1,\dots,\b_{n-1})\Eq(assa)
$$
with
$$
F(\b_1,\dots,\b_{n-1})~=~\sum_{\t\in T_n}\sum_{\ph_\t\in \Phi_\t}
\prod_{\{i,j\}\in E_\t}  |V(|x_i-x_j|)|e^{ -\b_{\phi_\t(\{i,j\})}{K_a}(|x_i-x_j|)}\Eq(asa)
$$
The sum over labels $\varphi_\t$  makes the function $F(\b_1,\dots,\b_{n-1})$ in the r.h.s. of \equ(asa)  symmetric in the variables $\b_1,\dots,\b_{n-1}$.
This
means that in the r.h.s. of \equ(assa) every variable $\b_i$ can be integrated from 0 to $\b$  as long as one  divides by $(n-1)!$. Hence we get
$$
|\phi_\b(x_1,\dots,x_n)|~\le ~{e^{nB\beta}\over (n-1)!}
\sum_{\t\in T_n}\sum_{\ph_\t\in \Phi_\t}
\prod_{\{i,j\}\in E_\t} \int_0^\b d\b_{\phi_\t(\{i,j\})} ~ |V(|x_i-x_j|)|e^{ -\b_{\phi_\t(\{i,j\})}{K_a}(|x_i-x_j|)}~~~
$$
$$
~~~~=~ {e^{nB\beta}\over (n-1)!}
\sum_{\t\in T_n}
\prod_{\{i,j\}\in E_\t} |V(|x_i-x_j|)|{(1-e^{-\b{K_a}(|x_i-x_j|)}) \over {K_a}(|x_i-x_j|)}\sum_{\ph_\t\in \Phi_\t}1~~
$$
$$
~= ~{e^{nB\beta}}\sum_{\tau\in T_n}\prod_{\{i,j\}\in E_\tau} \b|V(|x_i-x_j|)|{[1-e^{-\b{K_a}(|x_i-x_j|)}]\over \b{K_a}(|x_i-x_j|)}~~~~~~~~~~~~~~
$$
where of course  ${[1-e^{-\b{K_a}(|x_i-x_j|)}]\over \b{K_a}(|x_i-x_j|)}=1$ when ${K_a}(|x_i-x_j|)=0$.
Therefore
$$
|C_n(\b,\L)|~\le~{1\over |\L|}{1\over n!}\int_{\L}\,dx_1
\dots \int_{\L} dx_n\: |\phi_\b(x_1,\dots,x_n)| \le~
~{e^{nB\beta}\over n!} \sum_{\tau\in T_n}w_\t
$$
with
$$
w_\t~=~{1\over |\L|}\int_{\L}\,dx_1
\dots \int_{\L} dx_n\:
\prod_{\{i,j\}\in E_\tau}|V(|x_i-x_j|)|{[1-e^{-\b{K_a}(|x_i-x_j|)}]\over {K_a}(|x_i-x_j|)}
$$
We can estimate $w_\t$, for any $\t\in T_n$,  as follows. Let us set,
for $e=\{i,j\}\in E_\t$, $y_e=x_i-x_j$, then
$$
w_\t~\le
~
{1\over |\L|}\int_{\L}\,dx_1 \prod_{e\in E_\tau}\int_{\mathbb{R}^d} dy_e~ |V(|y_e|)| {[1-e^{-\b{K_a}(|y_e|)}]\over {K_a}(|y_e|)}\;=~
\left[\int_{\mathbb{R}^d} dx ~ |V(|x|)|{[1-e^{-\b{K_a}(|x|)}]\over {K_a}(|x|)} \right]^{n-1}
$$
Therefore we get
$$
|C_n(\b,\L)|~\le ~{e^{nB\beta}\over n!}\left[\int_{\mathbb{R}^d} dx ~ |V(|x|)|{[1-e^{-\b{K_a}(|x|)}]\over {K_a}(|x|)} \right]^{n-1}
\sum_{\tau\in T_n}1~= ~
$$
$$
~~~~~~=~
{e^{nB\beta}n^{n-2}\over n!}\left[\int_{\mathbb{R}^d} dx ~ \b |V(|x|)|{[1-e^{-\b{K_a}(|x|)}]\over \b {K_a}(|x|)} \right]^{n-1}
$$
where in the last line we have used the Cayley formula $\sum_{\tau\in T_n}1=n^{n-2}$ \cite{Ca}.
Finally, recalling definitions \equ(ka) and \equ(cstar), the bound \equ(A1) follows.

\vv
\\We conclude by proving 2), i.e. inequality \equ(A2).

\\If we use the constant $\overline{B}$ we have that, for any $\a\subset [n]$

$$
\sum_{\{i,j\}\subset \a} {V_a}(|x_i-x_j|)\ge -\overline{B}(|\a|-1)
$$
Thus, for the partition $\p\in \Pi(\t_k)$ of $[n]$,  we have that
$$
\sum_{\a\in \pi(\t_k)}U_a(\a)\ge -\overline{B}\sum_{\a\in \pi(\t_k)}(|\a|-1)
$$
Since $\t_k$ has exactly $n-k$ connected components, we have
$$
\sum_{\a\in \pi(\t_k)}(|\a|-1)=k
$$
So we get
$$
\sum_{\a\in \pi(\t_k)}U_a(\a)\ge -\overline{B}k
$$
Therefore, recalling that $U(\a)=U_a(\a)+\tilde U_a(\a)$ and that $\b_k-\b_{k+1}\ge 0$ for all $k\in [n-1]$, we get the following
bound for the exponent in the integrand of the r. h. s. of inequality \equ(ini)
$$
\exp\left\{ -\sum_{k=1}^{n-1}(\b_k-\b_{k+1})\sum_{\a\in \pi(\t_k)}U(\a)\right\}~\le~
\exp\left\{-\sum_{k=1}^{n-1}(\b_k-\b_{k+1})[-\overline{B} k+\sum_{\a\in \pi(\t_k)}\tilde U_a(\a)]\right\}
$$
Now observe that
$$
 \sum_{k=1}^{n-1}(\b_k-\b_{k+1})  k =\sum_{k=1}^{n-1}\b_k
$$
and, as before (see \equ(see))
$$
\sum_{k=1}^{n-1}(\b_k-\b_{k+1})\sum_{\a\in \pi(\t_k)}\tilde U_a(\a) \geq ~\sum_{\{i,j\}\in E_\t}\b_{\ph_\t(\{i,j\})}{K_a}(|x_i-x_j|)
$$
so that
$$
e^{ -\sum\limits_{k=1}^{n-1}(\b_k-\b_{k+1})\sum\limits_{\a\in \pi(\t_k)}U(\a)}~ \le~
e^{\sum\limits_{k=1}^{n-1}\b_k\overline{B}-\sum\limits_{\{i,j\}\in E_\t}\b_{\ph_\t(\{i,j\})}{K_a}(|x_i-x_j|)}~=
~ e^{-\sum\limits_{\{i,j\}\in E_\t}\b_{\ph_\t(\{i,j\})}[{K_a}(|x_i-x_j|)-\overline{B}\,]}
$$
\\Therefore plugging inequality above into \equ(ini) we get the estimate
$$
|\phi_\b(x_1,\dots,x_n)|~\le~
\int_{0}^\b d\b_1\int_0^{\b_1} d\b_2\dots \int_{0}^{\b_{n-2}}d\b_{n-1}\sum_{\t\in T_n}\sum_{\ph_\t\in \Phi_\t}
\prod_{\{i,j\}\in E_\t} |V(|x_i-x_j|)|~\times~~~~~~~~~~~~~~~~~~~~~
$$
$$
\times~e^{-\sum\limits_{\{i,j\}\in E_\t}\b_{\ph_\t(\{i,j\})}\Big[{K_a}(|x_i-x_j|)-\overline{B}\Big]}~~~~~~~~~~~~~~~~~~~~~~~~~~~~~~~~
$$
As before, we can  symmetrize in the r.h.s. the integration domain of the variables $\b_1,\dots,\b_{n-1}$ integrating all them between $0$ and $\b$ and dividing
by $(n-1)!$. We thus obtain
$$
|\phi_\b(x_1,\dots,x_n)|~= ~\sum_{\tau\in T_n}\prod_{\{i,j\}\in E_\tau}|V(|x_i-x_j|)|{[1-e^{-\b({K_a}(|x_i-x_j|)-\overline{B})}]\over ({K_a}(|x_i-x_j|)-\overline{B})}
~~~~~~~~~~~~~~~~~~~~~~~~~~~~~~~~~~~~~~~
$$
$$
~~~~~~~~~~~~~~= ~\left[{e^{\b\overline{B}}-1\over \b\overline{B}} \right]^{n-1}\sum_{\tau\in T_n}\prod_{\{i,j\}\in E_\tau}
\b|V(|x_i-x_j|)|{\b\overline{B}[1-e^{-\b({K_a}(|x_i-x_j|)-\overline{B})}]\over(e^{\b\overline{B}}-1) \b({K_a}(|x_i-x_j|)-\overline{B})}
$$
where the factor
${[1-e^{-\b({K_a}(|x_i-x_j|)- \overline{B})}]\over \b({K_a}(|x_i-x_j|)- \overline{B})}$ is set equal to 1 when  ${K_a}(|x_i-x_j|)-  \overline{B}=0$.
Therefore,
$$
|C_n(\b,\L)|
~\le~
{\left[{e^{\b\overline{B}}-1\over \b\overline{B}} \right]^{n-1}n^{n-2}\over n!}
\left[\int_{\mathbb{R}^d} dx~ \b|V(|x|)|{\b\overline{B}[1-e^{-\b({K_a}(|x|)-\overline{B})}]\over(e^{\b\overline{B}}-1) \b({K_a}(|x|)-\overline{B})} \right]^{n-1}
$$
Then, recalling definitions \equ(ka) and \equ(chat), the bound \equ(A2) follows.
\\$\Box$

\vv\vv
\\We stress that, in most of the cases, inequality \equ(A2) is stronger than \equ(A1). Let
$\b_0$ is the  solution of the equation ${(\exp\{\b\overline{B}\}-1)/ \b\overline{B}} =\exp\{\b B\}$. Then
the fact that \equ(A2) is stronger than \equ(A1)
 for all temperatures $\b\le \b_0$ follows from the following proposition.
\begin{pro}\label{pro3}
For all $|x|\le a$ and for all $\b \overline{B}\ge 0$ we have that
$$
{\b\overline{B}(1-e^{-\b [V(|x|)-V(a)]-\overline{B}})\over \b(V(|x|)-V(a)-\overline{B})(e^{\b\overline{B}}-1)}\le
{(1-e^{-\b [V(|x|)-V(a)]})\over \b[V(|x|)-V(a)]}\Eq(is1)
$$
\end{pro}
{\bf Proof}.
Let us denote shortly
$$
A=\b(V(|x|)-V(a))\Eq(iaia)
$$
By assumption we have that $A\ge 0$ for all $|x|\le a$.
Define the function $f: \mathbb{R}\to \mathbb{R}$ with
$$
f(y)=\cases{{y(e^{y-A}-1)\over (y-A)(e^y-1)}&if $y\neq 0, A$\cr\cr {1-e^{-A}\over A}&if $y=0$\cr\cr
{A\over e^A-1} & if $y=A$
}\Eq(effe)
$$
Let us check that $f(y)$ is, for all $A\ge 0$,  monotone decreasing in $\mathbb{R}$.
Since $f(y)>0$ for all $y\in\mathbb{R}$, we have that
$$
{d\ln f(y)\over dy}= {f'(y)\over f(y)}<0~~\Rightarrow~~ f'(y)<0
$$
But
$$
{d\ln f(y)\over dy}={d\over dy}\left[\ln y+\ln(e^{y-A}-1)-\ln(y-A)-\ln(e^y-1) \right]=
$$
$$
=~g(y-A)-g(y)~~~~~~~~~~~~~~~~~~~~~~~~~~~~~~~~~\Eq(truu)
$$
where
$$
g(y)=\cases{ {e^y\over e^y-1}-{1\over y} & if $y~\in \mathbb{R}\setminus\{0\}$\cr\cr
{1\over 2} & if $y~={0}$
}
$$
Thus the r.h.s. of \equ(truu) is negative if $g(y)$ is increasing.
We have
$$
g'(y)= {1\over y^2} -  {e^y\over (e^y-1)^2}
$$
so that
$$
g'(y)> 0 ~~~\Longleftrightarrow ~~ y^2< {(e^y-1)^2\over e^y}={e^{2y}-2e^y+1\over e^y}={e^y+ e^{-y} -2}
$$
I.e. $g(y)$  is increasing (and hence $f(y)$ is decreasing) if
$$
 y^2< {e^y+ e^{-y} -2}= 2\left[{e^y-e^{-y}\over 2}-1\right]= 2(\cosh(y)-1)=2\left[{y^2\over 2}+ {y^4\over 4!} +\dots\right]
$$
which is true for any $y$.

\\Thus the function $f(y)$ is decreasing and therefore, since $\b \overline{B} \ge 0$, we have
$$
f(\b \overline{B})\le f(0)
$$
which, recalling definitions \equ(iaia) and \equ(effe), coincides with the inequality \equ(is1).
$\Box$

\numsec=5\numfor=1
\section{New bounds for Lennard-Jones and Lennard-Jones type potentials. Comparison with previous bounds}\label{sec5}
In this final section  we perform a comparison between bounds \equ(A1) and \equ(A2) deduced from the Basuev Tree-graph identity and the recent
bounds for Lennard-Jones type potentials  given in \cite{MPS} and \cite{LP}.

\\Let us first introduce the standard definition of a Lennard-Jones type potential (a.k.a. inverse power low potential).
See e.g. \cite{Fi}, or definition 3.2.10 in \cite{Ru} or formula (4.2.2) in
\cite{ga}.
\begin{defi}\label{LenJ}
A   pair potential $V$ on  $\mathbb{R}^d$ is  of   Lennard-Jones type  if there exist positive constants $w$,
$r_1$, $r_2$, with $r_1\le r_2$, $C$, $C'$, $\e$   such that

$$
V(r)~\ge~ \cases{  {C\over r^{d+\e}} & if $r\le r_1$\cr\cr
 - w &  if $r_1\le r\le r_2$\cr\cr
 -{C'\over r^{d+\e}} & if $r\ge r_2$
}\Eq(condLJ)
$$
\end{defi}

\vv\vv
\\A Lennard-Jones type potential according to the definition above is clearly tempered and it is long known that it
is also stable (see e.g. \cite{Fi,Do, FR, Ru, ga}). Here below we prove that a Lennard-Jones
type potential is also Basuev according to Definition \ref{def9}.
\begin{pro}\label{pro2}
Let $V$ be a pair  potential on  $\mathbb{R}^d$ such that
there exist constants $w$,
$r_1$, $r_2$, with $r_1\le r_2$,  and  non-negative monotone decreasing functions
$\xi$, $\eta$ with domain in $[0,\infty)$ such that
$$
V(r)\ge\cases{\xi(r) & if $r\le r_1$\cr\cr
 - w &  if $r_1< r< r_2$\cr\cr
-\eta(r) & if $r\ge r_2$
}\Eq(condLJg)
$$
with
$$
\lim_{a\to 0} \x (a) a^d =+\infty \Eq(ult)
$$
and
$$
\int_{|x|\ge r_2} dx ~ \eta (|x|) <+\infty\Eq(long)
$$
Then $V$ is Basuev.
\end{pro}

\begin{proof} By Theorem \ref{basu1} we just need to show that there exists $a$ such that \equ(cb0) and \equ(cba) are satisfied. Fix $a\in (0,r_1)$,
let $\bar w=\max\{w, \h(r_2)\}$
$$
\bar \h(r)=\cases{\h(r) & if $r> r_2$ \cr\cr
\bar w &if $r\le r_2$}
$$
Then, by construction  $\bar \h$ is monotone decreasing  and such that $\int_{\mathbb{R}^{d}}dx ~ \bar \h(|x|) < \infty$.
Moreover by conditions \equ(condLJg) we have that
$$
V^-\le \bar \h
$$
Hence, recalling \equ(mua)  and considering also that,  since we took $a\in (0,r_1)$, by hypothesis $V^-(|x|)=0$ for all $|x|\le a$, we have
$$
\m(a)\le\sup_{n\in \mathbb{N},~(x_1,\dots, x_n)\in \mathbb{R}^{dn}\atop |x_i-x_j|>a,~|x_i|>a}\sum _{i=1}^n \bar \h(|x_i|)
$$
To  bound from above $\sum _{i=1}^n \bar \h(|x_i|)$, having in mind that  all particles are at mutual
distances greater than  $a$ and  at
distance greater than  $a$ from the origin, we proceed as follows.
We  draw for each $x_j$ a hypercube $Q_j$ with side-length $a/ 2\sqrt{d}$
(such that the maximal diagonal of $Q_j$ is $a/2$) in such a way that
$x_j$ is  a  vertex of the cube $Q_j$  and at the same time is the point of $Q_j$ farthest away from the origin $0$. Since any two
points among $x_1 ,\dots, x_n$ are at mutual distances $\ge
a$ the cubes so constructed do not overlap.  Furthermore, using the fact that $\bar\h$ is monotone decreasing,  we
have
$$
\h(|x_j|)\le {(4d)^{d\over 2}\over a^d}\int_{Q_j} dx~\h(|x|)
$$
Recall in fact that the cube $Q_j$ is chosen in such way that
$|x|\le |x_j|$ for all $x\in Q_j$.
Therefore
$$
\sum _{i=1}^n \bar \h(|x_i|)~\le~
{(4d)^{d\over 2}\over a^d}\sum_{i=1}^n\int_{Q_i} dx~\h(|x|)
=~{(4d)^{d\over 2}\over a^d}\int_{\cup_i Q_i}dx~ \bar\h(|x|)
\le~
{{(4d)^{d\over 2}\over a^d}}\int
_{\mathbb{R}^d} dx~ \bar\h(|x|)=~ {C_d\over a^d}
$$
where
$C_d=~ {(4d)^{d\over 2}\over a^d}\int_{\mathbb{R}^d}dx ~ \bar\h(|x|)$.
Hence
we get
$$
\mu(a)\le {C_d\over a^d}
$$
In view of  condition \equ(ult), we can always choose $a$ such that $\x(a)a^d > 2C_d$. Thus we get
$$
\x(a)a^d\ > 2C_d ~~~\Longrightarrow~~ \x(a)> 2{C_d\over a^d}~~~\Longrightarrow~~~V(a)> 2{C_d\over a^d} ~~~\Longrightarrow~~~ V(a)> 2\mu(a)
$$
$\Box$
\end{proof}
\vv
\begin{cor}\label{cor4}
A Lennard-Jones type potential according to Definition \ref{LenJ} is Basuev
\end{cor}
\begin{proof}
Just observe that setting $\x(|x|)={C/|x|^{d+\e}}$ and $\eta(|x|)=C'/|x|^{d+\e}$ a Lennard-Jones type potential satisfies the conditions
\equ(condLJg)-\equ(long). $\Box$
\end{proof}

\medskip

\\In what follows
we first compare  the bound given by \equ(A1) with the Morais-Procacci-Scoppola bound
\equ(radmru)  and show that bound \equ(A1) always beats bound \equ(radmru). As we have seen, if $B=\overline{B}$  or the temperature is not too low,
bound \equ(A2) beats  \equ(A1). We then consider the specific case of
the classical rescaled Lennard-Jones type potential $V(r)={1\over r^{12}}-{2\over r^6}$. In this specific case, due to known computational results
it is known that $\overline{B}$ and $B$ are, if not equal, then very close to each other  (i.e. $\overline{B}\le (1.001)B$).
So, in order to compare our results
in this case with those obtained    by de Lima and Procacci  in \cite{LP} we set, accordingly to them,  $\b=1$ and we show that
\equ(A2) strongly beats the de Lima and Procacci bound.

\subsection{Lennard-Jones type potentials}
As we have shown in Corollary \ref{cor4}, a Lennard-Jones type potential $V$ according to definition \ref{LenJ}  is
Basuev. I.e., if $B$ is the  stability constant of  $V$, then  there exists $a>0$ such that $V= V_a+ K_a$ with
$V_a$ stable and absolutely summable
with the same stability constant $B$ of the full $V$ and $K_a$ positive  supported in $[0,a]$
($V_a$ and $K_a$ defined as in \equ(va) and \equ(ka)).

\\Let us thus use this information to compute the (lower) bound given by Morais et al.
 for  the convergence radius, say $\tilde R$, of the Mayer series of a gas whose particles interact through a Lennard-Jones type potential $V$.
According to  Theorem \ref{teo2} this bound is as follows.
$$
\tilde R={1\over e^{{\b B}+1} \tilde C(\b)}\Eq(rad1)
$$
where
$$
\tilde C(\b) =\int_{|x|\le a} dx~ \left[1-e ^{-\b [V(|x|)-V(a)]}+\b V(a)\right]
+ \int_{|x|\ge a}dx~\b |V(|x|)|\Eq(bigger)
$$
Note that, in force of Corollary \ref{cor4}, we have replaced in the r.h.s. of \equ(rad1) the (generally larger) stability constant
$\tilde B$ of the absolutely summable potential $V_a$ (appearing in  the original Morais et al. bound \equ(radmru))
with the (generally smaller) stability
constant $B$ of the full potential $V$.

\\Let us compare this bound with the bound, say $R^*$, on the same radius derived from \equ(A1) obtained in Theorem \ref{teop}, which is
$$
R^*= {1\over e^{{\b B}+1} C^*(\b)}\Eq(rad2)
$$
with
$$
C^*(\b)~=~ \int_{|x|\le a} dx~ \b\, V(|x|)
{(1-e^{-\b [V(|x|)-V(a)]})\over \b[V(|x|)-V(a)]}+ \int_{|x|\ge a} dx~ \b\,|V(|x|)|=~~~~~~~~~~~~~~~~~~~~~~~~~~~~~~~~~~~~~
$$
$$
~~~~=~ \int_{|x|\le a} dx ~ \left[
{(1-e^{-\b [V(|x|)-V(a)]})}+ \b V(a){(1-e^{-\b [V(|x|)-V(a)]})\over \b[V(|x|)-V(a)]}\right]+ \int_{|x|\ge a} dx~ \b|V(|x|)|
$$
So we get
$$
\tilde C(\b) -C^*(\b)=\b V(a)  \int_{|x|\le a} dx~ \left[1- {1-e^{-\b [V(|x|)-V(a)]}\over \b[V(|x|)-V(a)]}\right]>0
$$
Indeed, recalling that $V(r)\ge V(a)$ for all $r\le a$ and setting $z=\b [V(|x|)-V(a)]$,
the function $1-{1- e ^{-z}\over z}$ is greater than zero for all $z>0 $. Hence we always  have that
$$
R^* > \tilde R
$$
I.e. the bound \equ(A1) is always better than the bound \equ(radmru).
\vv

\\Let us briefly consider also the case a  pair potential with an hard-core, i.e. a potentials $V_{hc}(r)$ such that $V_{hc}(r)=+\infty$ for $r\le a$ (with $a>0$).
 As mentioned in Section \ref{sec2}, a tempered pair potential with an hard-core is Basuev, and, as mentioned in Section \ref{secbas},
potentials with a hard-core can be treated with the Basuev tree
graph identity \equ(treeb) by introducing an auxiliary  potential $V(H)$.
For a potential $V_{hc}$ with hard-core as above
the factors $C^*(\b)$ and $\hat{C}(\b, \overline{B})$ appearing in bounds \equ(A1) and \equ(A2) become
$$
C^*(\b)= W_a(d)+ \b\int_{|x|\ge a} dx~ |V_{hc}(|x|)|\Eq(cstarhc)
$$
$$
\hat{C}(\b, \overline{B})=
{\b\overline{B}\over e^{\b\overline{B}}-1}\,W_a(d)+ \b\int_{|x|\ge a} dx~ |V_{hc}(|x|)| \Eq(chathc)
$$
where $W_a(d)$ the volume of the sphere of radius $a$ in $d$ dimensions.
Inserting \equ(cstarhc) into
\equ(A1), one gets the same lower bound for convergence radius  obtained in \cite{PU,MPS} for potentials with a hard-core  (see e.g. Theorem 3 in \cite{MPS}).
It is also clear that \equ(chathc) inserted into \equ(A2) yields an improved  lower bound for the convergence radius respect to that of Theorem 3 in \cite{MPS}.

\vv\vv
\subsection{\bf Lennard-Jones potential}
In this and the next subsection we consider
 the specific case of the three-dimensional  classical (rescaled) Lennard-Jones potential:
$$
V(r)={1\over r^{12}}-{2\over r^6}\Eq(LJ)
$$
We  denote by $B_{_{\rm LJ}}$  its stability constant.
\vv
\\We first compare the bound  stemming from Theorem \ref{teop} with the bound given in \cite{LP}  (see Theorem 3). To be
coherent with \cite{LP} we  set the  inverse temperature at the fixed value $\b=1$.

\\In \cite{LP} the authors show that,
by choosing $a=0.3637$,   the potential $V$ defined in \equ(LJ) can be written as
$V= V_a+ K_a$ with  $V_a(r)$, defined as in
\equ(va), stable with the same stability constant ${B_{_{\rm LJ}}}$ of the full Lennard-Jones potential $V$
and $K_a(r)= V(r)- V(a)$ for $r\in (0,a]$ and $K_a(r)=0$ elsewhere. The lower bound obtained in \cite{LP},
for the convergence radius,  say $R_{\rm LP}$, was
deduced from Theorem \ref{teo2} of Section \ref{sec2}. Namely. from inequality  \equ(radmru) with $\b=1$ one gets
$$
R_{\rm LP}= {1\over \tilde C(1)e^{{B_{_{\rm LJ}}}+1}}
$$
For the purpose of comparison with Basuev bound, observe that
$$
\tilde C(1) =\int_{\mathbb{R}^3} dx~ \left[ |e ^{- K_a(|x|)} -1|+ |V_a(|x|)|\right]   \ge ~~\int_{\mathbb{R}^3} dx~ |V_a(|x|)|
=~~~~~~~~~~~~~~~~~~~~~~~~~~~~
$$
$$
 = ~4\pi\int_0^{0.3637} dr~
{\left[{1\over |0.3637|^{12}}-{2\over |0.3637|^6}\right]r^2}~~+
~~ 4\pi\int_{0.3637}^\infty dr~ \left|{1\over r^{10}}-{2\over r^4}\right|~~\ge
$$
$$
\ge~~ 37444+ 12381~~\ge~~49825 ~~~~~~~~~~~~ ~~~~~~~~~~~~~~       ~~~~~~~~~~~~ ~~~~~~~~~~~~~~    ~~~~~~~~
$$
Therefore the De Lima-Procacci lower bound $R_{\rm LP}$ obtained in \cite{LP} can be at best ${e^{-({B_{_{\rm LJ}}}+1)}\over 49825 }$, i.e.
$$
R_{\rm LP}\le {e^{-({B_{_{\rm LJ}}}+1)}\over 49825 }\Eq(rlp)
$$

\\Now
 we use the  estimate
\equ(A2) (which in this case is tighter than \equ(A1)) to obtain an alternative lower bound  of the same
convergence radius. Since we want to use \equ(A2),
to compare  efficiently this bound with the previous one \equ(rlp) we need to estimate $\overline{B}_{_{\rm LJ}}$, which appears in
\equ(A2), in terms of $ B_{_{\rm LJ}}$. To do this
we take advantage here  of the   considerable amount of computational rigorous results on optimal Lennard-Jones clusters in the literature.
In particular
the values of  the global minima for
configuration with $n$ particles has been calculated up to $n=1610$ (see e.g. \cite{cambridge}
and reference therein). The tables in \cite{cambridge}
 show that $\bar {B}_n={n\over n-1}{B_n}$  (we recall that ${B_n}$, according to  \equ(bn),  is  the absolute value of the minimal energy of a
configuration with $n$ particles of a Lennard-Jones gas divided by $n$) is less than $\overline{B}_{1001}$ for all $n\le1000$.
So using these data one can conclude that $ {\overline{B}_{_{\rm LJ}}}=\sup_n\bar {B}_n\le {(1.001)} {B_{_{\rm LJ}}}$ where recall that
$B_{_{\rm LJ}}$ denotes the stability constant of the Lennard-Jones potential \equ(LJ).

\\Therefore, setting  $\b=1$
and $a=0.3637$ coherently with above, the lower  bound, say
$\hat R$, of the same convergence radius is, according to estimate \equ(A2),  given by
$$
\hat R~=~ {({1.001}) {B_{_{\rm LJ}}}
\over e\left(e^{{1001\over 1000} {B_{_{\rm LJ}}}}-1\right)}{1\over \hat C(1,\bar {B_{_{\rm LJ}}})} =
 h({B_{_{\rm LJ}}})\cdot{e^{-({B_{_{\rm LJ}}}+1)}\over \hat C(1,\bar {B_{_{\rm LJ}}})} \Eq(Bfins)
$$
with
$$
h(u)= {(1.001) u
\over \left(e^{ {u\over 1000}}-e^{-u}\right)} ~~~~~~~~~~~~~~{\rm for}~~u>0 \Eq(hb)
$$
and
$$
\hat{C}(1, \bar {B_{_{\rm LJ}}})= \int_{|x|\le a} dx~ \,V(|x|)
{\bar {B_{_{\rm LJ}}}(1-e^{- [V(|x|)-V(a)-\bar {B_{_{\rm LJ}}}]})\over (V(|x|)-V(a)-\bar {B_{_{\rm LJ}}})(e^{\bar {B_{_{\rm LJ}}}}-1)}+
\int_{|x|\ge a} dx~ |V(|x|)|
$$
First observe that the function $h(u)$ defined in \equ(hb) is increasing up to a value around $u=1000$ and it is known
that $8.61\le {B_{_{\rm LJ}}}\le 14.316$. The lower bound was first obtained in \cite{JI}  (see also Lemma 3 in \cite{SABS})
while the upper bound has been recently tightened
by Yuhjtman \cite{Y} (the previous best upper bound \cite{SABS} was ${B_{_{\rm LJ}}}\le 41.66$). Hence $h(u)$ is increasing
in the interval $[8.61, 14.316]$ and so $h(B_{_{\rm LJ}})\ge h(8.61)\ge 8.69$. Therefore we get from \equ(A2) that $\hat R$ is at worst
${8.69\over e^{{B_{_{\rm LJ}}}+1}\hat C(1, {\overline{B}_{_{\rm LJ}}})}$, i.e.
$$
\hat R~\ge~ {8.69\over e^{{B_{_{\rm LJ}}}+1}\hat C(1, {\overline{B}_{_{\rm LJ}}})}
$$
\\By bound \equ(is1) in Proposition \ref{pro3} we have
$$
\hat{C}(1,  {\overline{B}_{_{\rm LJ}}}) \le \hat{C}(1, {\overline{B}_{_{\rm LJ}}}=0)= \int_{|x|\le a} dx~ \,V(|x|)
{(1-e^{- [V(|x|)-V(a)]})\over V(|x|)-V(a)}+ \int_{|x|\ge a} dx~ |V(|x|)|
$$
Computing  explicitly $\hat{C}(1,  {\overline{B}_{_{\rm LJ}}}=0)$, we get
$$
\hat{C}(1,  {\overline{B}_{_{\rm LJ}}}=0)~~=~~ 4\pi\int_0^{0.3637} dx~ \left[{1\over |x|^{10}}-{2\over |x|^4}\right]
{1-e^{- \left[{1\over |x|^{12}}-{2\over |x|^6}-{1\over |0.3637|^{12}}+{2\over |0.3637|^6}\right]}\over
{1\over |x|^{12}}-{2\over |x|^6}-{1\over |0.3637|^{12}}+{2\over |0.3637|^6}}~+~~~~~~~~
$$
$$
~~~~+~~ 4\pi\int_{0.3637}^\infty dx~ \left|{1\over |x|^{10}}-{2\over |x|^4}\right|~\le~~0.823+ 12381.1~~\le~~12382\Eq(chat2)
$$
Hence we have that
$$
\hat R~\ge~ {8.69\over 12382\cdot e^{{B_{_{\rm LJ}}}+1}}~\ge~ {e^{-({B_{_{\rm LJ}}}+1)}\over 1425}\Eq(week)
$$
so that for  the ratio $R_{\rm LP}/\hat R$ we have
$$
{R_{\rm LP}\over \hat R}
~\le~
{{1\over 49825\cdot e^{{B_{_{\rm LJ}}}+1}}\over  {1\over 885\cdot e^{{B_{_{\rm LJ}}}+1}}}~=~{1425\over 49825}~\le~
{1\over 34,96}
$$
In conclusion we obtain that the bound \equ(A2) produces a
lower bound for the convergence radius which is, for the classical Lennard-Jones  potential,
nearly  35 times better than
the de Lima-Procacci lower bound for the same convergence radius.

\subsection{Optimal bound for Lennard-Jones potential}
\\The bound \equ(week) can  be strongly improved by trying to find the optimal $a$ which minimizes the factor $\hat C$ (which depends on $a$).
This optimal $a$  coincides
with  the $a$ as large as possible. This is because the first of the two terms in the r.h.s. of \equ(chat2) increases very slightly as $a$ increases,
 while
the second term in the r.h.s. of \equ(chat2) decreases very fast as $a$ increases. The total effect being
that
$\hat C(\b,{\overline{B}_{_{\rm LJ}}})$ decreases rapidly as $a$ increases.
Of course we cannot increase $a$ as we like.  According to Theorem \ref{basu1},  $a$ is constrained to satisfy the inequality $2\mu(a) \leq V(a)$.
Any $a$ satisfying this inequality will do and, due to the previous discussion, the larger is  $a$, the better.
To find this optimal $a$, we need an upper bound  as tight as possible  for $\mu(a)$. For this purpose
we can use the  bound recently obtained by Yuhjtman (see  \cite{Y}, Proposition 3.1 II).

\begin{pro}[Yuhjtman 2015]
Let $V$ be the Lennard-Jones potential defined in \equ(LJ), let $a$ be such that $0.6\le a\le 0.7$ and let $\m(a)$ be defined as in \equ(mua).
Then the following bound holds.
$$
\m(a)\le  {24.05\over a^3}
$$
\end{pro}


\\Thus the inequality $2\mu(a) \leq V(a)={1\over a^{12}}-{2\over a^6}$ is satisfied
if there exists $a\in  [0.6,0.7]$ such that
$$
2\cdot {24.05\over a^3}\le {1\over a^{12}}-{2\over a^6}
$$
which is true as soon as $a\le 0.6397$. Redoing  the calculations in \equ(chat2) replacing $a=0.3637$ with this new optimized value
$a=0.6397$ we get
$$
\hat{C}(1,  {\overline{B}_{_{\rm LJ}}}=0)~=~ 4\pi\int_0^{0.6397} dx~ \left[{1\over |x|^{10}}-{2\over |x|^4}\right]
{(1-e^{- \left[{1\over |x|^{12}}-{2\over |x|^6}-{1\over |0.6397|^{12}}+{2\over |0.6397|^6}\right]})\over
({1\over |x|^{12}}-{2\over |x|^6}-{1\over |0.6397|^{12}}+{2\over |0.6397|^6})}~~~+
$$
$$
~~+~~ 4\pi\int_{0.6397}^\infty dx~ \left|{1\over |x|^{10}}-{2\over |x|^4}\right|~\le~~2.5+ 61.63~~\le~~64.13
$$
With this bound for $\hat C$
we get our optimal lower bound for the convergence radius of the Mayer series of the Lennard-Jones gas at $\b=1$ which is
$$
\hat R~\ge~ {8.69\,e^{-({B_{_{\rm LJ}}}+1)}\over 64.13 }~\ge~ {e^{-({B_{_{\rm LJ}}}+1)}\over
7.4 }\Eq(strong)
$$
which is greater than the previous best lower bound \equ(rlp) given in \cite{LP} by  factor $6,7\times 10^4$.
\vv
\\Finally, it is worth to remark  that, using the recent upper bound for the stability constant of the Lennard-Jones potential
$B_{_{\rm LJ}}\le 14.316$ obtained by one of us \cite{Y},
the lower bound of convergence radius of the Mayer series of the Lennard-Jones gas given in \equ(strong)
is, in absolute terms,
$$
6,7\times 10^4\times e^{41.66- 14.316}  \ge 5\times 10^{16}
$$
times better than the previous known best bound available in the literature, i.e. the bound
\equ(rlp) given in \cite{LP} where the upper bound $B_{_{\rm LJ}}\le 41.66$ was used.

\renewcommand{\theequation}{A.\arabic{equation}}
\setcounter{equation}{0}  
\section*{Appendix A. Proof of Theorem \ref{basu0}}  
\\We begin by obtaining an  expression of the Ursell coefficients $\phi_\b(X)$ (see Section \ref{secbas}, Definition \ref{ursell}, formula \equ(2.1s))
 which is alternative
to the standard one  given in Lemma \ref{conne}, equation \equ(2.4s).

\begin{lem}\label{Basu}
The following identities hold for all non-empty $X\subset [n]$.

\begin{equation}\label{Basueva}
\phi_\b(X)= \sum_{\pi\in \Pi(X)}  (-1)^{|\pi|-1}(|\pi|-1)! e^{-\b\sum_{\a \in \pi} U(\a)}
\end{equation}
\end{lem}

\\{\bf Proof}.
The identity (\ref{Basueva}) can be deduced
directly via the M\"obius inversion formula. We first recall that
the set of partitions $\Pi(X)$  of a finite set $X$ is a lattice, where the partial order $\le$ is the usual partition refinement.   Namely, if $\pi,\pi'\in \Pi(X)$ then we say that
$\pi$ is a refinement of $\pi'$ and we write
$\pi\le \pi'$ if any block $\a$ of $\pi$ is such that $\a\subset \a'$ for some block $\a'$ of  $\pi'$. Hence the
maximum in the lattice
$\Pi([n])$ is the one-block partition $\vartheta=\{[n]\}$
while the minimum is the $n$-block partition $\theta=\{\{1\},\{2\},\dots\{n\}\}$.

\\Let $f:\Pi(X)\to \mathbb{R}$  be defined
such that
$$
f(\pi)=\prod_{\a\in \pi} \phi_\b(\a)
$$
Let us further define $g:\Pi(X)\to \mathbb{R}$ such that,
$$
g(\pi)= \prod_{\a\in\pi} e^{-\b U(\a)}
$$
Then, recalling
that $\vartheta\in \P(X)$ is the one-block (maximal) partition of $X$,
$g(\vartheta)=e^{-\b U(X)}$ and $f(\vartheta)=\phi_\b(X)$. So with these notations \equ(2.1s)
rewrites to
$$
e^{-\b U(X)}= g(\vartheta)=\sum_{\pi\le \vartheta}f(\pi)
$$
Therefore,  by M\"obius inversion formula, (see e.g. \cite{Rota64})
$$
\phi_\b(X)=f(\vartheta)= \sum_{\pi\le \vartheta}\m(\pi,\vartheta)g(\pi)
$$
where $\m(\pi,\vartheta)$ is the M\"obius function of the lattice of set partitions.
It is now well known that $\m(\pi,\vartheta)=
(-1)^{|\pi|-1}(|\pi|-1)!$ (see again \cite{Rota64} or also  \cite{Wi}, formula 5C.20), whence
(\ref{Basueva}) follows.
$\Box$
\vskip.15cm
\\{\bf Remark}.
Lemma \ref{Basu} implies the non trivial identity
\begin{equation}\label{ntr}
\phi_\b([n])~=~\sum_{\pi\in \Pi([n])}  (-1)^{|\pi|-1}(|\pi|-1)!
e^{-\b\sum_{\a \in \pi} U(\a)}=\sum_{g\in G_n}\prod_{\{i,j\}\in E_g} \left[e^{-\b V_{ij}}-1\right]
\end{equation}

\\In particular, if $K^n_k$ is the number of partitions of the set $[n]$ with exactly $k$ blocks (i.e. $K^n_k=|\Pi_k([n])|$),
then by setting $\beta =0$ in (\ref{ntr}),  it follows that
\begin{equation}\label{ntriv}
\sum_{k=1}^n (-1)^{k-1}(k-1)!K^n_k =0
\end{equation}
The representation (\ref{ntr}) of the Ursell coefficients is the starting point in order to prove
the  tree graph identity given  in 1979 by Basuev \cite{Ba2}.  Moving on towards this goal,
let us introduce some further notation.
We set shortly
$$\int_0^\b d\b_1\int_0^{\b_1}d\b_2\cdots \int_0^{\b_{n-2}}d\b_{n-1}\equiv\int_S d\bm \b$$

\\Let  $A=\{\a_1,\dots,\a_k\}$ be
a family of disjoint non-empty subsets of $[n]$. An element $\a$ of the family $A$ is called a block
and $k=|A|$ is the cardinality of the family.
We
let $q(A)=\{\s\subset A: |\s|=2\}$. Namely $q(A)$ is  the set of all unordered pairs of blocks of $A$.
Note that  $\s\in q(A)$ implies that  $\s=\{\a,\a'\}$ with $\a\subset [n]$,
$\a'\subset [n]$ both nonempty  and
such that $\a\cap \a'=\emptyset$. Moreover, for $\s\in q(A)$ and $A$ family of disjoint of non-empty subsets of $[n]$, we denote
\begin{equation}\label{Ws}
W_\s= \sum_{\{i,j\}\subset [n]\atop i\in \a,~j\in \a'} V_{ij} 
\end{equation}
\begin{equation}\label{uhat}
\hat U(A)=\sum_{\s\in q(A)}W_\s 
\end{equation}
Note that (\ref{Ws}) and (\ref{uhat}) jointly with  Definition \ref{algpot} imply that
$\hat U(A)= ~U(\cup_{i=1}^k \a_i)-\sum_{i=1}^kU(\a_i)$ and also $U(\a)+U(\a')+W_\s= U(\a\cup \a')$.
Note also that if $|A|=1$ (i.e. if $A$ contains just one block) then $\hat U(A)=0$, since $q(A)$ is empty.

\\Let
$\g=\{\a_1,\dots, \a_k\}$ be   a partition of $[n]$  and  let $\s=\{\a_i,\a_j\}\in q(\g)$, then we denote by
$P_\s \g$ the new partition of $[n]$  given by $P_\s \g=(\g\setminus\{\a_i,\a_j\})\cup \{\a_i\cup \a_j\}$.
Namely  $P_\s \g \in \P([n])$ is the new partition of $[n]$ given by
$$
P_\s \g= \{\a_1,\dots, \a_{i-1}, \a_{i+1},\dots , \a_{j-1}, \a_{j+1},\dots, \a_k, \a_i\cup \a_j\}
$$
Note  that $|P_\s \g|=|\g|-1$.

\\Finally, if $\g\in \P([n])$ then  we denote by  $\Pi_l(\g)$  the set of all partitions of $\g$ with fixed cardinality $l$.
We stress that   $\g\in  \Pi([n])$ is a partition of $[n]$. So
an element $\p\in \Pi_l(\g)$,  is in general {\it not} a partition of $[n]$ but rather a ``partition of a  partition". Namely,
 a block $A\in \pi$ 
is in general a disjoint family of sets of $[n]$.

\\With the notations above, for any partition $\g\in \P([n])$,   we define
\begin{equation}\label{gene}
\psi_\b(\g)=  \sum_{k=1}^{|\g|} (-1)^{k-1}(k-1)!\sum_{\pi\in \P_k(\g)}e^{-\b\sum_{A\in \pi} \hat U(A)} 
\end{equation}
Note that
\begin{equation}\label{C1}
\psi_\b(\g)=1 ~~~~~~~~~~~~~{\rm if}~~~~~~|\g|=1 
\end{equation}
\\Observe also that, by (\ref{ntriv}), we have, for any $\g\in \P([n])$
\begin{equation}\label{b0}
\psi_{\b}(\g)=0,\ \mbox{ for }\beta =0   
\end{equation}
Observe finally that, with these notations, if we denote by $\g_0$ the partition of $[n]$  given by

\noindent \mbox{$\g_0=\{\{1\},\{2\},\dots, \{n\}\}$}, then by equation (\ref{ntr}) we have that
$$
\psi_\b(\g_0)=\phi_\b([n])=\sum_{g\in G_n}\prod_{\{i,j\}\in E_g} \left[e^{-\b V_{ij}}-1\right]
$$

\begin{lem}\label{phiC} For every  $\g\in \P([n])$ with  $|\g|\ge 2$ we have
\begin{equation}\label{pic}
\psi_\b(\g)=-\int_0^\b d\b'\sum_{\s\in q(\g)}W_\s e^{-\b' W_\s} \psi_\b(P_\s \g) 
\end{equation}
\end{lem}
\\{\bf Proof}. By definition (\ref{gene}) we have that $\psi_\b(\g)$ is differentiable   as a function of $\b$, hence, using also (\ref{b0})  we can write
$$
\psi_\b(\g)=\int_0^\b d\b'~ {\partial}_{\b'} \psi_{\b'}(\g) =
 - \int_0^\b d\b'~ \left[\sum_{k=1}^{|\g|} (-1)^{k-1}(k-1)!\sum_{\pi\in \P_k(\g)}\sum_{A\in \pi} \hat U(A)e^{-\b'\sum_{A\in \pi} \hat U(A)}\right]
$$
Hence we just need to show that
\begin{equation}\label{show}
\sum_{\s\in q(\g)}W_\s e^{-\b W_\s} \psi_\b(P_\s \g)= \sum_{k=1}^{|\g|} (-1)^{k-1}(k-1)!\sum_{\pi\in \P_k(\g)}\sum_{A\in \pi}
\hat U(A)e^{-\b\sum_{A\in \pi} \hat U(A)} 
\end{equation}
We have
$$
\sum_{\s\in q(\g)}W_\s e^{-\b W_\s} \psi_\b(P_\s \g)=   \sum_{k=1}^{|\g|-1} (-1)^{k-1}(k-1)! \sum_{\s\in q(\g)}W_\s \sum_{\tilde \pi\in \P_k(P_\s \g)}
e^{-\b\left[\sum_{A\in \tilde\pi} \hat U(A)+W_\s\right]}
$$
Let $\pi$ be a partition of $\g$. We write $\s\sqsubset \pi$ if there is $A\in \pi$ such that $\s\subset A$.  With this notation we can write
$$
\sum_{\tilde \pi\in \P_k(P_\s \g)}
e^{-\b\left[\sum_{A\in \tilde\pi} \hat U(A)+W_\s\right]}=\sum_{\pi\in \P_k(\g)\atop \s\sqsubset \pi} e^{-\b \sum_{A\in \pi} \hat U(A)}
$$
Hence
$$
 \sum_{\s\in q(\g)}W_\s \sum_{\tilde \pi\in \P_k(P_\s \g)}
e^{-\b\left[\sum_{A\in \tilde\pi} \hat U(A)+W_\s\right]}= \sum_{\s\in q(\g)}W_\s
\sum_{\pi\in \P_k(\g)\atop \s\sqsubset \pi} e^{-\b \sum_{A\in \pi} \hat U(A)}=
$$
$$
=\sum_{\pi\in \P_k(\g)} e^{-\b \sum_{A\in \pi} \hat U(A)}~ \sum_{\s\in q(\g)\atop \s\sqsubset \pi}W_\s =
\sum_{\pi\in \P_k(\g)} e^{-\b \sum_{A\in \pi} \hat U(A)}~ \sum_{A\in \pi} \hat U(A)
$$
Note that $ \sum_{\s\in q(\g):~ \s\sqsubset \pi}W_\s=\sum_{A\in \pi} \hat U(A)$ because both sides are the sum
of  $W_\s$  over all pairs of blocks in $\g$ such that both blocks are in the same block $A$ in $\pi$.
Hence we get
$$
\sum_{\s\in q(\g)}W_\s e^{-\b W_\s} \psi_\b(P_\s \g)= \sum_{k=1}^{|\g|-1} (-1)^{k-1}(k-1)! \sum_{\pi\in \P_k(\g)}
e^{-\b \sum_{A\in \pi} \hat U(A)}~ \sum_{A\in \pi} \hat U(A)
$$
When $k=|\g|$  there is only one partition $\pi$  in $|\g|$ blocks and $\pi$ is such that each block $A\in \pi$ contains just one element and thus
$\hat U(A)=0$. Hence we can change the upper limit $|\gamma|-1$ in the sum over $k$ to $|\gamma|$ in the r.h.s. of equation above so that it becomes (\ref{show}).
$\Box$

\begin{defi}
Let $k\in [n]$, we define a class $\mathfrak{S}_k$ of  sequences  $(\s_1,\dots,\s_k)$ as follows:
$(\s_1,\dots,\s_k)\in \mathfrak{S}$ if
$\s_1\in q(\g_{0})$ and, for $i\ge 2$,  $\s_i\in q(P_{\s_{i-1}}\cdots P_{\s_1}\g_0)$. Given $(\s_1,\dots,\s_k)\in \mathfrak{S}$ we also
denote, for all $i\in [k]$, $\g_{i}=P_{\s_{i}}P_{\s_{i-1}}\cdots P_{\s_1}\g_0$. We finally denote shortly $\mathfrak{S}_{n-1}=\mathfrak{S}$.
\end{defi}

\begin{cor}\label{BaTGI}
Let  $V$ be a pair interaction  in $[n]$.
Then the Ursell coefficient $ \phi_\b ([n])$
satisfies the following identity.

$$
\phi_\b([n])=\psi_\b(\g_0)=~(-1)^{n-1}\int_S d\bm \b
\sum_{(\s_1,\dots,\s_{n-1})\in \mathfrak{S}}W_{\s_1}\cdots
W_{\s_{n-1}}e^{-\sum_{i=1}^{n-1}\b_i W_{\s_i}}
$$
\end{cor}
\\{\bf Proof}. Just
iterate $n-1$ times  the equation (\ref{pic}) to $\psi_\b(\g_0)$
and observe that by construction $|P_{\s_{n-1}}\dots P_{\s_1}\g_0|=1$,
so
by (\ref{C1}),
$$
\psi_\b(P_{\s_{n-1}}\cdots P_{\s_1}\g_0)=1
$$

\begin{lem}
Let $k\in [n]$ and let $(\s_1,\dots,\s_k)\in \mathfrak{S}_k$, then
\begin{equation}\label{pipp}
\sum_{i=1}^k W_{\s_i}=\sum_{\a\in \g_k}U(\a) 
\end{equation}
\end{lem}
\begin{proof} By induction on $k$. For $k=1$ the identity (\ref{pipp}) is true. Indeed, supposing  $\s_1=\{i,j\}$ and hence
$$
\g_1=\left\{\{1\},\dots, \{i-1\},\{i+1\},\dots,  \{j-1\},\{j+1\},\dots,\{n\},\{i,j\}\right\}
$$
we have
$$
W_{\s_1} = V_{ij}= U(\{i,j\})=  U(\{i,j\})+ \sum_{l\in [n]: \atop l\neq i,~l\neq j} U(\{l\})=\sum_{\a\in \g_1}U(\a)
$$
Assuming thus (\ref{pipp}) true for $k$, let us prove that (\ref{pipp}) also holds for $k+1$. Let $\s_{k+1}=\{\a',\a''\}$
with $\a'\in \g_k$ and $\a''\in\g_k$.
Then we have
$$
\sum_{i=1}^{k+1} W_{\s_i}= \sum_{i=1}^{k} W_{\s_i}+ W_{\s_{k+1}} =\sum_{\a\in \g_k}U(\a)+ W_{\s_{k+1}}=
$$
$$
=
\sum_{\a\in \g_k\atop \a\neq\a',~\a\neq \a''}U(\a) + U(\a')+U(\a'')+W_{\s_{k+1}}= \sum_{\a\in \g_k\atop \a\neq\a',~\a\neq \a''}U(\a) +
U(\a'\cup\a'')= \sum_{\a\in \g_{k+1}}U(\a)
$$
$\Box$\end{proof}
\begin{cor}\label{co2} Let $\s_1,\dots,\s_{n-1}$ be such that $\s_i\in q(\g_{i-1})$ and $\g_{i}=P_{\s_{i}}\g_{i-1}$.
The following identity holds:
\begin{equation}\label{idem}
\sum_{i=1}^{n-1}\b_i W_{\s_i}=\sum_{k=1}^{n-1}(\b_{k}-\b_{k+1})\sum_{\a\in \g_k}U(\a) 
\end{equation}
where we have put  $\b_n=0$.
\end{cor}
\begin{proof} We have that $\b_i=\sum_{k=i}^{n-1} (\b_{k}-\b_{k+1})$ and thus, using (\ref{pipp})
$$
\sum_{i=1}^{n-1}\b_i W_{\s_i}=\sum_{i=1}^{n-1}\sum_{k=i}^{n-1} (\b_{k}-\b_{k+1}) W_{\s_i}= \sum_{k=1}^{n-1}(\b_{k}-\b_{k+1})\sum_{i=1}^{k}  W_{\s_i} =
 \sum_{k=1}^{n-1}(\b_{k}-\b_{k+1}) \sum_{\a\in\g_k}U(\a)
$$
$\Box$
\end{proof}

\\To state the following lemma we need to introduce some further notation. Recall that $\GG_n$ is the set of all graphs with vertex set $[n]$ and,
given $g\in \GG_n$, $\pi(g)$ denotes the partition of $[n]$ induced by $g$ such that a block $\a\in \p(g)$ is
formed by the vertices of  a connected component of $g$.

\\Given $\g\in \P([n])$ and  $\s=\{\a,\a'\}\in q(\g)$, and given $\l=\{i,j\}\in {\rm E}_n$,
we write $\l\lhd \s$ if $i\in \a$ and $j\in \a'$.

\begin{lem}\label{l5}
Let  $\FF_n$ be the set of all pairs $((\s_1,\dots,\s_{n-1}), (\l_1,\dots,\l_{n-1}))$ with $(\s_1,\dots,\s_{n-1})\in \mathfrak{S}$ and,
for $i=1,\dots,n-1$, $\l_i\in {\rm E}_n$ and  $\l_i\lhd \s_{i}$.
Let $\TT_n$ be the set of all pairs
$\bm\t=(\t, \ph_\t)$ where $\t$ is a tree with vertex set $[n]$ and $\ph_\t$ is a labeling of the edges of $\t$
(i.e. a bijection $\ph_\t: E_\t\to [n-1]: \{i,j\}\mapsto \ph_\t(\{i,j\})$ with $E_\t$ denoting the edge set of $\t$).
Then there is a bijection between $\SS_n$ and $\TT_n$.
\end{lem}
\begin{proof} Given the pair $((\s_1,\dots,\s_{n-1}), (\l_1,\dots,\l_{n-1}))\in\FF_n$,  $\bm\t=(\l_1,\dots,\l_{n-1})$ is a set of labeled edges in $[n]$.
Let us prove that $\bm\t$ is a tree with vertex set  $[n]$. Consider the graph $\bm\t_k=(\l_1,\dots,\l_k)$ (with $k\le n-1$). By construction  the connected components of  $\bm\t_k$ are the blocks forming
$\gamma_k=P_{\s_k}\cdots P_{\s_1}\gamma_0$, i.e.
$\p(\t_k)= \gamma_k$. Hence $\bm\t=\bm\t_{n-1}$ has only one connected component, \mbox{i. e.} $\bm\t$ is a connected graph with $n-1$ edges, thus a tree.
The labels of the edges are the subindexes, so $\bm\t=(\l_1,\dots,\l_{n-1})\in \TT_n$.

\\Conversely let $\bm\t=(\t, \ph_\t)\in \TT_n$ with labeled edges $(\l_1,\dots,\l_{n-1})$ and let $\t_k$ be
the graph in $\GG_n$ obtained from  $\t$ by erasing all edges with labels greater than $k$.
Then a unique sequence  $\g_1,\dots, \g_{n-1}$  of partitions
of $[n]$ is determined by posing $\g_k=\pi(\t_k)$. This sequence  $\g_1,\dots, \g_{n-1}$ uniquely determines a sequence  $\s_1,\dots, \s_{n-1}$
such that
$\s_i\in q(\g_i)$ and $\g_i=P_{\s_i}\g_{i-1}$ and by construction $\l_i\lhd \s_i$. Hence we have constructed a function that associates to an edge-labeled
tree $\bm\t\in \TT_n$ with
labeled edges $(\l_1,\dots,\l_{n-1})$
a unique pair $((\s_1,\dots, \s_{n-1});(\l_1,\dots,\l_{n-1}))\in \FF_n$. We call $(\s^\t_1,\dots, \s^{\bm\t}_{n-1})$ this unique sequence determined
by the edge-labeled tree $\bm\t=(\l_1,\dots,\l_{n-1})$.
$\Box$
\end{proof}

\vv\vv

\vv
\\{\bf Proof of \equ(treeb)}.  By Corollary \ref{BaTGI} we have
$$
\phi_\b([n])=\psi_\b(\g_0)~=~(-1)^{n-1}\int_S d\bm \b
\sum_{(\s_1,\dots,\s_{n-1})\in \mathfrak{S}}W_{\s_1}\cdots
W_{\s_{n-1}}e^{ -\sum_{i=1}^{n-1}\b_i W_{\s_i}}~~~~~~~~~~~~~~~~~~~~~~~~~
$$
$$
~~~~~~~~~~~~~~~~~ = ~(-1)^{n-1}\int_S d\bm \b
\sum_{(\s_1,\dots,\s_{n-1})\in \mathfrak{S}}
\sum_{\l_1\in {\rm E}_n\atop\l_1\lhd \,\s_1}V_{\l_1}\dots
\sum_{\l_{n-1}\in {\rm E}_n\atop \l_{n-1} \lhd\, \s_{n-1}}V_{\l_{n-1}}e^{ -\sum_{i=1}^{n-1}\b_i W_{\s_i}}
$$
$$
~~~~~~~~~~~= ~(-1)^{n-1}\int_S d\bm \b
\sum_{(({\s_1,\dots,\s_{n-1});(\l_1,\dots,\l_{n-1}))\in \FF_n}}\left[\prod_{k=1}^{n-1} V_{\l_k}\right]~e^{ -\sum_{i=1}^{n-1}\b_i W_{\s_i}}
$$
$$
~= ~(-1)^{n-1}\int_S d\bm \b
\sum_{{\bm \t}=(\l_1,\dots,\l_{n-1})\in \TT_n}\left[\prod_{k=1}^{n-1} V_{\l_k}\right]~~
e^{ -\sum_{i=1}^{n-1}\b_i W_{\s^{\bm \t}_i}}
$$
then \equ(treeb) follows by Corollary \ref{co2}, formula (\ref{idem}). 

\renewcommand{\theequation}{B.\arabic{equation}}
\setcounter{equation}{0}  
\section*{Appendix B. Proof of Theorem \ref{basu1}}  

We start  recalling that
$$
V^-(r) = {1\over 2}\Big[|V(r)|-V(r)\Big]
$$
The thesis is trivial if $V^-=0$ (i.e. if $V$ is purely repulsive). So we may assume that  $V^-\neq0$.
For any $(x_1,\dots,x_n)\in \mathbb{R}^{nd}$ and any $i\in [n]$, let
$$
E_i(x_1,\dots,x_n)=\sum_{j\in [n]:\,j\neq i}V(|x_i-x_j|)
$$
so that
$$
U(x_1,\dots,x_n)=  E_1(x_1,\dots,x_n)+U(x_2,\dots,x_n)
$$
Let now
$(x_1,\dots,x_n)\in \mathbb{R}^{nd}$ be a configuration in which
 there is a particle, say in position $x_1$ (without loss of generality), such that $E_1(x_1,\dots,x_n)\ge 0$. Then
we have
$$
U(x_1,\dots,x_n)\ge  U(x_2,\dots,x_n)
$$
therefore, since ${1\over n}<{1\over n-1}$,
$$
-{1\over n}U(x_1,\dots,x_n)<  -{1\over n-1}U(x_2,\dots,x_n)
$$
Thus we have that the configuration $(x_2,\dots,x_n)$ produce a value $-U(x_2,\dots,x_n)/(n-1)$ which is nearer to $B$ than
$-U(x_1,\dots,x_n)/n$. Whence we can look for minimal energy configurations $(x_1,\dots, x_n)$ limiting ourselves to those configurations in which
the energy per particle $E_i(x_1,\dots,x_n)$ is negative for all $i\in [n]$.

\\Now let us consider the system of
particles interacting via the pair potential $V_a$ defined in \equ(va) and let us
assume that conditions \equ(cb0) and  \equ(cba) holds. Note first that, due to condition \equ(cb0),
$V_a^-(|x|)=\max\{0, - V_a(|x|)\}=V^-(|x|)$.
Consider then  a configuration $(x_1,\dots,x_n)$ such that there exists $\{i,j\}\subset [n]$ such that
$|x_i-x_j|\le a$, thus there is at least a particle, (which, without loss of generality,  we can assume to be the particle indexed by 1 at position $x_1=0$),
which has the maximum number of particles among $x_2,\dots,x_n$ at distance less than or equal to $a$.
Say that the number of these particles close to $x_1$ less or equal to $a$ is $l$ ($l\ge 1$ by assumption). The energy $E_1$ of the particle at position $x_1$
can thus be estimated as follows.
$$
E_1(x_1,\dots,x_n)\ge l V(a) -\sum_{k\in [n]\atop|x_k|>a} V^-(|x_k|)
$$
To control the sum $\sum_{k} V^-(|x|)$ observe that we are supposing that each particle has at most $l$ other particles at distance less or equal than $a$.
Thus take the  $k\in [n]$ such that $V^-(|x_k|)$ is maximum. Again, without loss of generality we can suppose $k=2$. In the sphere with center $x_2$ and radius $a$ there are at most
$l+1$ particles (the particle at position $x_2$ plus at most $l$ other particles).
Hence
$$
\sum_{k\in [n]\atop |x_k|>a} V^-(|x_k|)\le  (l+1)V^-(|x_2|)+\sum_{k\in [n]\atop |x_k|>a, |x_k-x_2|>a} V^-(|x_k|)
$$
Iterating  we get
\begin{equation}\label{lpu}
\sum_{k\in [n]\atop |x_k|>a} V^-(|x_k|)\le (l+1)\sum_{k\in [n]\atop  |x_i-x_j|>a} V^-(|x_k|)
\end{equation}
where in the sum in the r.h.s of (\ref{lpu}) all pairs of particles are at  distance greater than $a$ to each other. Therefore, recalling definition \equ(mua)
we have that
$$
\sum_{k\in [n]\atop  |x_i-x_j|>a} V^-(|x_k|)\le \m(a)
$$
and hence
$$
E_1(x_1,\dots,x_n)\ge l V(a) -(l+1)\m(a)
$$
so we have $E_1>0$ whenever
$$
V(a)>{l+1\over l} \mu(a)
$$
Using assumption \equ(cba) and since ${l+1\over l}\le 2$ we get
$$
E_1(x_1,\dots,x_n)>0
$$
In conclusion,  if a configuration $(x_1,\dots,x_n)$ is such that some particles are at distance less or equal than $a$,
then there is at least one  particle  whose energy is positive. Hence the minimal energy configurations for $V_a$ must be searched among
those configurations in which all particles are at distance greater than $a$ from each other. But for these configurations $V_a= V$ which
implies that $V_a$ and $V$, if stable, have the same stability constant $B$ (and similarly also the same $\overline{B}$). Finally, it is easy to see that
 $V$ is  stable. Just observe that for any configuration $(x_1,\dots,x_n)$ for which particles are at distance greater than $a$ from each other we have
$$
U_a(x_1,\dots,x_n)={1\over 2}  \sum_{i=1}^n\sum_{j\in [n]\atop j\neq i} V(|x_i-x_j|)\ge-{1\over 2}n\m(a)
$$
which implies that $V_a$ (and a fortiori $V$) are stable with stability constant $B\le { \m(a)\over 2}$. 
\vv

\section*{Acknowledgments}
The authors are very grateful to two anonymous referees whose comments, remarks and suggestions greatly helped to  improve
the presentation of the present paper.

\\B.N.B.L. and A.P.  have been partially supported by the Brazilian  agencies
Conselho Nacional de Desenvolvimento Cient\'{\i}fico e Tecnol\'ogico
(CNPq) and  Funda{\c{c}}\~ao de Amparo \`a  Pesquisa do Estado de Minas Gerais (FAPEMIG - Programa de Pesquisador Mineiro).
 S.Y. has been partially supported by Coordena\c{c}\~ao de Aperfei\c{c}oamento de Pessoal de N\'ivel Superior (CAPES)
and Departamento de Matemática UFMG. A.P. and S.Y. would like to thank the Instituto Nacional de Matemática Pura e Aplicada (IMPA) for kind hospitality
during December 2014.

\end{document}